\documentclass{article}
\usepackage{booktabs} 

\usepackage{amsmath,amsthm,mathtools,amssymb}
\usepackage[linesnumbered,ruled,vlined]{algorithm2e}

\SetCommentSty{mycommfont}

\usepackage[numbers,sort]{natbib}
\usepackage[colorlinks,citecolor=blue,linkcolor=blue]{hyperref}
\usepackage[dvipsnames]{xcolor}
\definecolor{bgred}{RGB}{240,192,208}
\definecolor{bggreen}{RGB}{208,255,192}
\definecolor{bgblue}{RGB}{202,218,250}
\definecolor{myred}{cmyk}{0.1, 0.75, .9, 0}
\definecolor{myblue}{cmyk}{.95, 0.45, 0, 0}
\definecolor{mygreen}{cmyk}{.8, 0.2, .8, 0}
\usepackage{enumitem}

\usepackage{scalerel,stackengine}
\newcommand\pig[1]{\scalerel*[5pt]{\big#1}{%
  \ensurestackMath{\addstackgap[1.5pt]{\big#1}}}}

\usepackage{fullpage}
\usepackage{cleveref,autonum}

\usepackage{subcaption}
\usepackage{tikz}
\usetikzlibrary{calc,shadows}
\tikzset{
  dot/.style = {circle, fill, minimum size=#1,inner sep=0pt, outer sep=0pt},
  dot/.default = 3pt,
}

\theoremstyle{definition}
\newtheorem{thm}{Theorem}
\newtheorem{df}{Definition}
\newtheorem{prop}{Proposition}
\newtheorem{lm}{Lemma}
\newtheorem{ex}{Example}
\newtheorem{coro}{Corollary}
\newtheorem{remark}{Remark}

\newcommand{\cF}{\mathcal{F}}
\newcommand{\cJ}{\mathcal{J}}
\newcommand{\cL}{\mathcal{L}}
\newcommand{\C}{\mathcal{C}}

\newcommand{\Rbar}{\underline{\mathbb{R}}}
\newcommand{\vn}{\varnothing}
\renewcommand{\mid}{:}

\DeclareMathOperator*{\argmin}{arg\,min}
\DeclareMathOperator*{\argmax}{arg\,max}

\usepackage{array}
\newcolumntype{C}{>{$}c<{$}}

\setcitestyle{authoryear}

\title{Properties of Path-Independent Choice Correspondences and\\ 
Their Applications to Efficient and Stable Matchings}

\author{Keisuke Bando, Kenzo Imamura, Yasushi Kawase}

\usepackage{authblk}
\author[1]{Keisuke Bando}
\affil[1]{Keio University, Japan}
\author[2]{Kenzo Imamura}
\author[2]{Yasushi Kawase}
\affil[2]{University of Tokyo, Japan}
\date{}

\begin{document}

\maketitle

\begin{abstract}
Choice correspondences are central to decision-making processes, particularly when indifferences (ties) arise. While tie-breaking can transform a choice correspondence into a choice function, it often introduces inefficiencies, highlighting the importance of directly addressing choice correspondences. 
This paper introduces a novel notion of path‐independence (PI) for choice correspondences, extending the existing concept of PI for choice functions. 
Intuitively, a choice correspondence is PI if any consistent tie‐breaking produces a PI choice function.
This new notion yields several important properties. 
First, PI choice correspondences are \emph{rationalizabile}, meaning they can be represented as the maximization of a utility function. This extends a core feature of PI in choice functions.
Second, we demonstrate that the set of choices selected by a PI choice correspondence for any subset forms a \emph{generalized matroid}. 
This property reveals that PI choice correspondences exhibit a nice structural property.
Third, we establish that choice correspondences rationalized by \emph{ordinally concave} functions inherently satisfy the PI condition. This aligns with recent findings that a choice function satisfies PI if and only if it can be rationalized by an ordinally concave function.
Building on these theoretical foundations, we explore stable and efficient matchings under PI choice correspondences. Specifically, we investigate \emph{constrained efficient} matchings, which are efficient (for one side of the market) within the set of stable matchings. Under responsive choice correspondences, such matchings are characterized by cycles.
However, this cycle-based characterization fails in more general settings.
We demonstrate that when the choice correspondence of each school satisfies both PI and monotonicity conditions, a similar cycle-based characterization is restored. 
These findings provide new insights into the matching theory and its practical applications.
\end{abstract}

\section{Introduction}
Choice is a fundamental concept in economics, central to understanding and modeling the behavior of economic agents. It plays a key role in both theoretical and empirical analysis. The combinatorial choice problem, which involves selecting a subset of elements from a set, has diverse applications in market design, particularly in matching theory. For example, it arises when a school decides which applicants to admit or when a firm selects a group of workers to hire. Such choice behaviors are typically formalized using choice functions, which identify the optimal subset from a given set of available options. 

In practice, however, choice behavior is often better represented by a \emph{choice correspondence} rather than a choice function. A choice correspondence allows multiple subsets to be selected from the same set of available options. This scenario frequently occurs in settings where ties or indifference exist.
For example, a school’s priority ranking may include ties (e.g., applicants with identical test scores), or affirmative action policies may account for distributions across applicant types (e.g., gender, race, socioeconomic status). Similarly, firms may only partially observe worker attributes (e.g., education level, major, job-related skills), making it impossible to differentiate among candidates with identical observable characteristics. While tie-breaking can convert a choice correspondence into a choice function, this process often introduces inefficiencies or may compromise desirable properties. Consequently, directly addressing choice correspondences is crucial for preserving these properties and ensuring robust analysis.

Our focus is on the properties of choice correspondences, with particular attention to \emph{path-independence (PI)}. PI is a fundamental concept for choice functions as it ensures consistency in selection. Formally, a choice function is PI if the chosen subset from any set remains the same regardless of whether the selection is made in one step or in multiple stages---first by partitioning the set into smaller subsets, selecting from each, and then choosing again from the union of those selections. PI choice functions offer several advantages: PI guarantees rationalizability and is equivalent to being rationalizable by an \emph{ordinal concave} utility function, which is a natural form of discrete concavity. In matching theory, PI choice functions ensure the existence of stable matchings, a central concept in the field.

In this paper, we extend the concept of PI to apply to choice correspondences in a natural way. 
For choice functions, PI is known to be equivalent to the conjunction of \emph{substitutability} and \emph{irrelevance of rejected contracts (IRC)}~\citep{aizerman1981general}. \citet{sotomayor1999three} proposed extensions of substitutability and IRC for choice correspondences.  
Consequently, one possible extension of PI for choice correspondences is to consider the conjunction of these extended notions of substitutability and IRC. 
However, this approach has certain limitations.
To address these points, we introduce a novel notion of PI for choice correspondences. Our definition generalizes the classical notion of PI from choice functions while strengthening the conjunction of substitutability and IRC for choice correspondences. Intuitively, a choice correspondence satisfies PI if any consistent tie-breaking results in a PI choice function.
In the following paragraphs, we explain the advantages of our proposed notion of PI over substitutability and IRC for choice correspondences.

First, our PI guarantees rationalizability (\Cref{thm:rationalize}): any PI choice correspondence can be represented as the maximization of some utility function. This result extends a key feature of PI choice functions. While it is established that PI choice functions are rationalizable~\citep{yang2020rationalizable}, we further demonstrate that our PI choice correspondences are likewise rationalizable. Rationalization is essential in economics because agents are typically assumed to make decisions consistent with a utility function or preference relation. Therefore, identifying the rationalization underlying observed choice behavior provides the intellectual foundation for nearly all economic analyses.\footnote{See, for example, Chapter~1 of the textbook by  \citet{mas1995microeconomic}.}
Theoretically, we demonstrate rationalizability by extending a relationship between PI choice functions and closure operators~\citep{johnson1996algebraic,koshevoy1999choice} to choice correspondences. 
In contrast, the conjunction of substitutability and IRC for choice correspondences does not necessarily ensure rationalizability (see \Cref{ex:PIunion}).

Second, any PI choice correspondence exhibits a desirable combinatorial structure. Specifically, we prove that for any subset of alternatives, the set of chosen outcomes selected by a PI choice correspondence forms a \emph{generalized matroid (g-matroid)} (\Cref{thm:gmatroid}). This structure provides several benefits. Notably, the outcome of a PI choice correspondence, including tie-breaking, can be computed efficiently (\Cref{prop:PILADcomp} and \Cref{thm:computeC}).
In contrast, for choice correspondences  satisfying the conjunction of substitutability and IRC, the set of chosen outcomes does not necessarily exhibit a matroidal structure, and the outcome including time-breaking cannot be computed efficiently (\Cref{rem:exp}).

Third, we show that any choice correspondence rationalized by an ordinally concave function satisfies PI (\Cref{thm:oc-PI}). This finding aligns with recent results indicating that a choice function is PI if and only if it can be rationalized by an ordinally concave function~\citep{yokote2024rationalizing}. Thus, a PI choice correspondence not only admits a rationalization but is also closely linked to this natural form of discrete concavity. This property is particularly useful in applications. To verify that a given choice correspondence satisfies our conditions, it suffices to show that it can be rationalized by a discrete concave function within this class.  Various techniques developed in the field of discrete convex analysis facilitate this verification; see \Cref{sec:application} for detailed discussion.

Building on these theoretical foundations, we examine stable and efficient matchings under PI choice correspondences. In particular, we focus on \emph{constrained efficient matchings}---those that are efficient within the set of stable matchings. 
The conjunction of substitutability and IRC in choice correspondences guarantees the existence of stable matchings~\citep{che2019weak}. Since our PI condition is stronger than these, it also guarantees existence. Specifically, a stable matching can be obtained by applying the \emph{deferred acceptance algorithm} to choice functions derived via tie-breaking. 
However, while this method guarantees stability, it may not yield an efficient stable matching (see \Cref{ex:dominated}). Under the standard responsive choice correspondences, \citet{erdil2008s} characterized constrained efficient matchings using cycles and provided an algorithm to find them.
Nevertheless, this cycle-based characterization does not necessarily hold in more general settings, such as acceptant and substitutable correspondences~\citep{erdil2019efficiency}. We demonstrate that when each school’s choice correspondence satisfies both PI and a monotonicity condition, the cycle-based characterization is restored (\Cref{thm:main}). This monotonicity extends an important condition in matching theory, known as \emph{the law of aggregate demand (LAD)}, to choice correspondences. Formally, a choice correspondence satisfies LAD if any consistent tie-breaking results in a choice function that satisfies LAD.

Our characterization applies to real-life matching markets, particularly those with diversity concerns. Many choice functions---such as \emph{quotas} \citep{abdulkadirouglu2003school}, \emph{reserves} \citep{hafalir2013effective,ehlers2014school}, and \emph{overlapping reserves} \citep{sonmez2022affirmative}---have been proposed. While most of these models assume strict priorities, our framework accommodates weak priorities (i.e., ties) and provides cycle-based characterizations of constrained efficient matchings. This ensures that a constrained efficient matching can be found in polynomial time in each scenario. Additionally, \citet{erdil2019efficiency} introduced a choice correspondence based on reserves and provided a similar characterization of constrained efficient matchings. Our results offer a structural understanding of their approach: their choice correspondence satisfies both PI and LAD.

\subsection{Related Work}

PI for choice functions was first introduced by \citet{plott1973path}. Under PI choice functions, it is known that stable matchings exist \citep{roth1984stability,aygun2013matching, blair1988}. Recently, \citet{yokote2024rationalizing} showed that PI choice functions can be characterized through the rationalization of ordinal concavity. If a choice function is rationalized by an M${}^\natural$-concave function, it satisfies both PI and LAD~\citep{FujishigeTamura2006,MurotaYokoi2015}. Furthermore, any choice correspondence that can be rationalized by these functions also satisfies PI and LAD.

\citet{johnson1996algebraic} and \citet{koshevoy1999choice} pointed out the relationship between PI choice functions and (finite) \emph{convex geometries}.
A convex geometry is a combinatorial structure that generalizes the concept of convexity in Euclidean spaces to more abstract settings. 
Formally, a convex geometry consists of a finite set paired with a closure operator that satisfies the anti-exchange property~\citep{edelman1985convex}. This structure induces a PI choice function through the \emph{extreme operator}. For further details, see Chapter~5 in the book by \citet{GratzerWehrung2016}. 
We generalize some of the results obtained in \citet{johnson1996algebraic} and \citet{koshevoy1999choice} for PI choice functions to PI choice correspondences. In particular,  we utilize these results to prove the rationalizability of a PI choice correspondence.

Several studies have been conducted on efficient and stable matchings under choice correspondences.
\citet{erdil2008s} characterized constrained efficient matchings under responsive choice correspondences, showing that a stable matching is constrained efficient if and only if it does not admit any stable improvement cycle.
\citet{erdil2019efficiency} and \citet{erdil2022corrigendum} analyzed constrained efficient matchings under acceptant and substitutable correspondences. In that setting, constrained efficient matchings may admit cycles (PSIC), making them difficult to fully characterize.
Under our conditions, however, we can generalize the characterization provided by \citet{erdil2008s}.
Our PI strengthens substitutability, whereas our LAD is weaker than acceptance. Therefore, a straightforward comparison with \citet{erdil2019efficiency} is not possible.

Our study also relates to the efficient allocation of indivisible goods under constraints. Specifically, a constraint on allocations can be represented by a choice correspondence that returns all feasible subsets.
\citet{suzuki2018efficient,STYYZ2023} generalized the top trading cycles mechanism that preserves Pareto efficiency, individual rationality, and group strategy-proofness for any distributional constraint representable by an M-convex set on the vector of the number of students assigned to each school, given an initial endowment.
\citet{IK2024,IK2024b} provided necessary and sufficient conditions for constraints to guarantee the existence of desired mechanisms.

Finally, our results have important implications for real-life applications in market design. In matching with diversity concerns, various choice functions---such as quotas \citep{abdulkadirouglu2003school}, reserves \citep{hafalir2013effective,ehlers2014school} and overlapping reserves \citep{sonmez2022affirmative}---have been proposed, typically assuming strict priorities. Using our framework, we can accommodate weak priorities (i.e., ties). In each case, we provide a characterization of constrained efficient matchings in terms of cycles (PSIC).

\subsection{Organization of the Paper}
The remainder of this paper is organized as follows: \Cref{sec:preliminaries} provides the necessary definitions used throughout this study.
\Cref{sec:PIchoice} introduces PI choice correspondences and examines their key properties.
\Cref{sec:constrainedPE} applies these theoretical foundations to explore stable and efficient matchings under PI choice correspondences, offering a characterization of constrained efficient matchings.
The applicability of our model to practical settings is discussed in \Cref{sec:application}.
Finally, \Cref{sec:conclusion} concludes with a discussion of our findings and suggestions for future research directions. 

\section{Preliminaries}\label{sec:preliminaries}

We denote the set of real numbers by $\mathbb{R}$, and the set of all positive real numbers by $\mathbb{R}_{++}$.
Additionally, we write $\Rbar$ to denote $\mathbb{R}\cup\{-\infty\}$.
Let $\mathbb{Z}_+$ represent the set of nonnegative integers.
For a set $X$ and an element $i$, we define $X + i = X \cup \{i\}$ and $X - i = X \setminus \{i\}$. Additionally, we define $X + \emptyset = X$ and $X - \emptyset = X$.

Let $I=\{i_1,i_2,\dots,i_n\}$ be a finite set of elements (students).
A \emph{choice function} $C\colon 2^I\to 2^I$ is a function that satisfies $C(X)\subseteq X$ for all $X\in 2^I$.
Here, $C(X)$ is interpreted as the most preferred subset from an available set $X$.
Such a set is not uniquely determined when preferences involve ties or indifferences.
In this paper, we focus on 
a \emph{choice correspondence}, which is a set-valued function $\C\colon 2^I\rightrightarrows 2^I$ that assigns to each $X\in 2^I$ a collection $\C(X)\subseteq 2^X$.
For each $X,Y\in 2^I$, assume that it is possible to determine whether $Y \in \C(X)$ in constant time.\footnote{For choice functions, it is standard to assume that one is provided with a choice oracle which directly returns $C(X)$ for any query $X\in 2^I$. However, for choice correspondences, $\C(X)$ may potentially have an exponential number of sets. Thus, we assume only the availability of a membership oracle. In \Cref{sec:computation}, we present a reduction from a membership oracle to a choice oracle for PI choice functions.}

In the following, we introduce important concepts used in this paper.

\subsection{Matroids and Generalized Matroids}
We will observe that PI choice correspondences exhibit a remarkable connection to the combinatorial structures of matroids and g-matroids.
Therefore, we start by introducing these two concepts.

A nonempty family of subsets $\cF\subseteq 2^I$ is a \emph{matroid} if, 
(i) $X\subseteq Y\in\cF$ implies $X\in\cF$, and 
(ii) for any $X,Y\in\cF$ with $|X|<|Y|$, there is $e\in Y\setminus X$ such that $X+e\in\cF$. 
A simple example of a matroid is the family $\{I'\in 2^I\mid |I'|\le q\}$, where $q$ is a nonnegative integer. This is known as a \emph{uniform matroid} of rank $q$.
A family $\cL\subseteq 2^I$ is called a \emph{laminar} if, for any $X,Y\in\cL$, we have $X\cap Y=\emptyset$, $X\subseteq Y$, or $X\supseteq Y$. For a laminar family $\cL\subseteq 2^I$ and $q\colon \cL\to\mathbb{Z}_{+}$, the family $\{I'\in 2^I\mid |I'\cap L|\le q_L~(\forall L\in\cL)\}$ is a matroid, which is called a \emph{laminar matroid}.
For a bipartite graph $G=(I,J; E)$, the family \(\{I'\in 2^I\mid \text{there exists a matching in $G$ that covers $I'$}\}\) is a matroid, known as a \emph{transversal matroid}.

A nonempty family of subsets $\cF\subseteq 2^I$ is called a \emph{generalized matroid (g-matroid)}\footnotemark{} if, for any $X,Y\in\cF$ and $e\in X\setminus Y$, there is $e'\in (Y\setminus X)\cup\{\emptyset\}$ such that $X-e+e'$ and $Y+e-e'$ are in $\cF$~\citep{Tardos1985}. 
\footnotetext{A g-matroid is also referred to as an \emph{M${}^\natural$-convex family} because the corresponding set of $0$--$1$ vectors is an M${}^\natural$-convex set as a subset of $\mathbb{Z}^I$~\citep{murota2016}.}
Alternatively, g-matroid can be characterized by another property~\citep{MS1999}: for any $X,Y\in\cF$ and $e\in X\setminus Y$, it holds that
(i) $X-e+e'\in\cF$ for some $e'\in (Y\setminus X)\cup\{\emptyset\}$, and 
(ii) $Y+e-e'\in\cF$ for some $e'\in (Y\setminus X)\cup\{\emptyset\}$.

One of the properties of a g-matroid $\cF$ is that any set $X\in\cF$ that is not maximum size can have an element added to it to form another set in the g-matroid. This property will be used later.
\begin{prop}\label{prop:gmatroid_nonmaximal}
For any g-matroid $\cF\subseteq 2^I$ and $X\in\cF$, if $|X|<\max\{|Y|\mid Y\in\cF\}$, then there is an element $i\in I$ such that $X+i\in\cF$.
\end{prop}
\begin{proof}
Let $Y^*\in\argmin\{|X\bigtriangleup Y|\mid Y\in\cF,\ |Y|>|X|\}$.
Since $|Y^*| > |X|$, there must be some element $e \in Y^* \setminus X$. 
Then, by the definition of g-matroid, there is $e'\in (X\setminus Y^*)\cup\{\emptyset\}$ such that $Y^*-e+e'\in\cF$.
As $|X\bigtriangleup(Y^*-e+e')|<|X\bigtriangleup Y^*|$, it follows that $e'=\emptyset$ and $|Y^*-e|=|X|$.
Consequently, $X+e=Y^*\in\cF$.
\end{proof}

\subsection{Properties of Choice Functions}
A choice function $C$ is called \emph{path-independent (PI)} if it satisfies $C(X\cup X')=C(C(X)\cup X')$ for any $X,X'\subseteq I$.
A choice function $C$ is called \emph{substitutable (SUB)} if $C(X)\cap X'\subseteq C(X')$ for any $X'\subseteq X\subseteq I$. Additionally, a choice function $C$ satisfies \emph{irrelevance of rejected contracts (IRC)} if $C(X')=C(X)$ for any $X,X'\subseteq I$ with $C(X)\subseteq X'\subseteq X$.
It is known that a choice function satisfies PI if and only if it satisfies both SUB and IRC~\citep{aizerman1981general}.
A choice function $C$ satisfies \emph{law of aggregate demand (LAD)}\footnotemark{} if $X'\subseteq X\subseteq I$ implies $|C(X')|\le |C(X)|$.
A choice function $C$ is called \emph{acceptant} if there exists a nonnegative integer $q$ such that $|C(X)|=\min\{|X|,\,q\}$ for every $X\in 2^I$. Clearly, each acceptant choice function satisfies LAD.
\footnotetext{This notion is also referred to as \emph{cardinal monotonicity}~\citep{Alkan2002} or \emph{size monotonicity}~\citep{AlkanGale2003}.}

A choice function $C$ is called \emph{rationalizable} if there is a utility function $u\colon 2^I \to \Rbar$ such that $\{C(X)\}=\argmax\{u(Y) \mid Y \subseteq X\}$ for any $X\in 2^I$.
Throughout this paper, we consider only utility functions $u\colon 2^I\to\Rbar$ that satisfy $u(\emptyset)=0$.
This ensures that $\argmax\{u(Y)\mid Y\subseteq X\}\ne\emptyset$ for any $X\in 2^I$.
We say that a utility function $u$ is \emph{unique-selecting} if $\argmax\{u(Y)\mid Y\subseteq X\}$ is a singleton for all $X\in 2^I$.
A utility function $u$ induces a choice function only if it is unique-selecting.
It is known that a choice function is rationalizable if and only if it satisfies the \emph{strong axiom of revealed preference (SARP)}~\citep{yang2020rationalizable}.

Next, we provide important classes of utility functions.
A utility function $u$ is said to be \emph{M${}^\natural$-concave} if, for any $X, X' \subseteq I$ and $i \in X \setminus X'$, there exists $j \in (X' \setminus X) \cup \{\emptyset\}$ such that $$u(X)+u(X')\le u(X-i+j)+u(X'+i-j).$$
We say that a utility function $u$ is associated with a \emph{weighted matroid} if it can be expressed as 
\begin{align}
u(X)=\begin{cases}
v(X)    & \text{if }X\in\cF,\\
-\infty & \text{if }X\notin\cF,
\end{cases}
\end{align}
where $v$ is an additive function (i.e., $v(X)=\sum_{i\in X}v(\{i\})~(\forall X\in 2^I)$) and $\cF$ is a matroid.
It is not difficult to verify that every function associated with a weighted matroid is M${}^\natural$-concave.
A utility function $u$ is called \emph{laminar concave} if it can be expressed as
\begin{align}
u(X)=\sum_{L\in\cL} \varphi_L(|X\cap L|),
\end{align}
where $\cL\subseteq 2^I$ is a laminar family and $\varphi_L$ is a univariate concave function for each $L\in\cL$.
Every laminar concave function is known to be M${}^\natural$-concave~\cite[Note 6.11]{murota2003}.

A utility function $u$ is called \emph{ordinal concave}\footnotemark{} if, for any $X, X' \subseteq I$ and $i \in X \setminus X'$, there exists $j \in (X' \setminus X) \cup \{\emptyset\}$ such that:
(i) $u(X) < u(X - i + j)$, 
(ii) $u(X') < u(X' + i - j)$, or
(iii) $u(X) = u(X - i + j)$ and $u(X') = u(X' + i - j)$.
\footnotetext{
The notion of ordinal concavity is equivalent to \emph{semistrictly quasi M${}^\natural$-concavity}. For more details, see the literature~\citep{murota2003quasi,murota2003,FarooqShioura2004,ChenLi2021,fujishige2024note}.}
A utility function $u$ satisfies \emph{size-restricted concavity} if, for any $X, X' \subseteq I$ with $|X|>|X'|$, there exists $i \in X \setminus X'$ such that:
(i) $u(X) < u(X - i)$, 
(ii) $u(X') < u(X' + i)$, or
(iii) $u(X) = u(X - i)$ and $u(X') = u(X' + i)$.

\citet{yokote2024rationalizing} characterized PI and LAD through the concepts of ordinal concavity and size-restricted concavity.
\begin{thm}[\citet{yokote2024rationalizing}]\label{thm:fPI}
A choice function is PI if and only if it is rationalizable by a utility function satisfying ordinal concavity.
Furthermore, a choice function is PI and LAD if and only if it is rationalizable by a utility function satisfying ordinal concavity and size-restricted concavity.
\end{thm}

It is known that a choice function associated with an M${}^\natural$-concave function is PI and LAD~\citep{MurotaYokoi2015,FujishigeTamura2006}.\footnotemark{}
This can also be justified by the fact that any M${}^\natural$-concave function satisfies both ordinal concavity and size-restricted concavity~\citep{yokote2024rationalizing}.

\footnotetext{\citet{MurotaYokoi2015} also proved that any unique-selecting \emph{quasi M${}^\natural$-concave} function induces a choice function that is PI and LAD, where quasi M${}^\natural$-concavity is a concept weaker than M${}^\natural$-concavity.}

\subsection{Properties of Choice Correspondences}\label{sec:correspondence}

A choice correspondence $\C\colon 2^I\rightrightarrows 2^I$ is said to be \emph{rationalizable} if there exists a utility function $u\colon 2^I \to \Rbar$ such that $\C(X)=\argmax\{u(Y) \mid Y \subseteq X\}$ for every $X\in 2^I$.
We provide a characterization of this concept by extending SARP in \Cref{sec:rationalize}.
For a family of subsets $\cF\subseteq 2^I$ with $\emptyset\in\cF$, the choice correspondence defined by $\C(X)=\{Y\subseteq X\mid Y\in\cF\}$ is rationalizable.
This type of choice correspondence is especially useful for modeling constraints without imposing priorities.

\citet{sotomayor1999three} introduced the \emph{substitutability} of a choice correspondence $\C$ as follows:
\begin{description}
    \item[(SC$^{1}_{\text{ch}}$)] For any $X_1, X_2\in 2^I$ with $X_1 \supseteq X_2$ and any $Z_1 \in \C(X_1)$, there exists $Z_2 \in \C(X_2)$ such that $X_2 \cap Z_1 \subseteq Z_2$.
    \item[(SC$^{2}_{\text{ch}}$)] For any $X_1, X_2\in 2^I$ with $X_1 \supseteq X_2$ and any $Z_2 \in \C(X_2)$, there exists $Z_1 \in \C(X_1)$ such that $X_2 \cap Z_1 \subseteq Z_2$.
\end{description}
\citet{sotomayor1999three} also introduced the \emph{IRC} condition of a choice correspondence $\C$: For any $X,Y,Y'\in 2^I$, if $Y \in \C(X)$ and $Y \subseteq Y' \subseteq X$, then $Y \in \C(Y')$.
It has been shown that a stable matching exists if the choice correspondence of each school satisfies substitutability and IRC (see \Cref{sec:matching} for more details). On the other hand, a choice correspondence may not be rationalizable even if it satisfies substitutability and IRC (see \Cref{ex:PIunion}).
Finally, we define the acceptance of a choice correspondence $\C$. 
A choice correspondence $\C$ is called \emph{acceptant} if there exists a nonnegative integer $q$ such that for every $X\in 2^I$ and for every $Y\in\C(X)$, it holds that $|Y|=\min\{|X|,\,q\}$.


\section{PI Choice Correspondences}\label{sec:PIchoice}
In this section, we define the central concept of this
paper, the \emph{path-independent (PI) choice correspondence}, and present its fundamental properties.

Let $w\colon I\to\mathbb{R}$ be a weight function.
We denote $w(X)=\sum_{i\in X}w(i)$ for each $X\in 2^I$.
A weight function $w$ is called \emph{unique maximizing (UM)} if, for every nonempty $\mathcal{X}\subseteq 2^I$, the set $\argmax_{X\in\mathcal{X}}w(X)$ is a singleton (i.e., $w(X)\ne w(X')$ for any distinct $X,X'\in 2^I$).
For any UM weight $w$, define $C^w$ to be the choice function such that $\{C^w(X)\}=\argmax_{Y\in\C(X)} w(Y)~(\forall X\in 2^I)$.
Intuitively, $C^w$ represents the outcome of applying a tie-breaking rule to $\C$, consistently selecting the subset with the highest weight.

\begin{df}\label{def:PI}
A choice correspondence $\C$ is PI if, for any UM weight $w$, the choice function $C^w$ satisfies PI.
\end{df}

Similarly, we define \emph{LAD} of a choice correspondence
as follows.
\begin{df}
A choice correspondence $\C$ is LAD if, for any UM weight $w$, the choice function $C^w$ satisfies LAD.
\end{df}

\begin{table}[htb]
\centering
\begin{tabular}{C|CCCCC}
\toprule
X           & \C_0(X)            &\C_{1}(X)       & \C_{2}(X)     & \C_{3}(X)         & \C_{4}(X)          \\\midrule
\emptyset   & \emptyset        &\emptyset       & \emptyset     & \emptyset         & \emptyset          \\
\{a\}       & \{a\}            &\{a\}           & \{a\}         & \{a\}             & \emptyset,\{a\}    \\
\{b\}       & \{b\}            &\{b\}           & \{b\}         & \{b\}             & \emptyset,\{b\}    \\
\{c\}       & \{c\}            &\{c\}           & \{c\}         & \{c\}             & \emptyset,\{c\}    \\
\{a, b\}    & \{a\},\{b\}      &\{a,b\}         & \{a\}         & \{b\}             & \emptyset,\{a,b\}  \\
\{a, c\}    & \{a\},\{c\}      &\{a,c\}         & \{a\}         & \{a\}             & \emptyset,\{a,c\}  \\
\{b, c\}    & \{b\},\{c\}      &\{b\},\{c\}     & \{b,c\}       & \{c\}             & \emptyset,\{b,c\}  \\
\{a, b, c\} & \{a\},\{b\},\{c\}&\{a,b\},\{a,c\} & \{a\},\{b,c\} & \{a\},\{b\},\{c\} & \emptyset,\{a,b,c\}\\
\bottomrule
\end{tabular}
\caption{Examples of choice correspondences}\label{tab:exChoice}
\end{table}

\Cref{tab:exChoice} provides some examples of choice correspondences on $I=\{a,b,c\}$. 
$\C_0$ and $\C_1$ satisfy PI and LAD. 
In contrast, $\C_2$ satisfies PI but not LAD.
To see that $\C_2$ fails LAD, note that $C_2^w(\{b,c\})=\{b,c\}$ while $C_2^w(\{a,b,c\})=\{a\}$ for a UM weight $w$ with $w(a)>w(b)+w(c)$.
Likewise, $\C_3$ satisfies LAD but not PI.
To see that $\C_3$fails not PI by $C_3^w(\{a,b\})=\{b\}$ while $C_3^w(\{a,b,c\})=\{a\}$ for a UM weight $w$ with $w(a)>w(b)>w(c)$.
Finally, $\C_4$ fails to satisfy both of PI and LAD; this can be verified by considering the UM weight $w$ with $(w(a),w(b),w(c))=(-1,2,-4)$ for which $C_4^w(\{a\})=\emptyset$, $C_4^w(\{a,b\})=\{a,b\}$, and $C_4^w(\{a,b,c\})=\emptyset$.

\subsection{Rationalizability}

In general, a choice correspondence that satisfies substitutability and IRC is not necessarily rationalizable, unlike the case for PI choice functions. The following example illustrates this fact.

\begin{ex}\label{ex:PIunion}
Consider the choice correspondence $\C_{4}$ given in 
Table \ref{tab:exChoice}. It can be expressed as the union of two PI choice functions, $C^{(1)}$ and $C^{(2)}$, where $C^{(1)}(X)=X$ and $C^{(2)}(X)=\emptyset$ for all $X\in 2^I$.
Thus, $\C_4$ satisfies substitutability and IRC (see \Cref{lm:PIunion} in \Cref{sec:sub_PI}).
We already have seen that $\C_{4}$ does not satisfy PI.
Moreover, $\C_4$ is not rationalizable because $\C_4(\{a\})=\{\{a\},\emptyset\}$ implies that the utilities of $\{a\}$ and $\emptyset$ are equal, whereas $\C_4(\{a,b\})=\{\{a,b\},\emptyset\}$ implies that the utility of $\{a\}$ is strictly smaller than that of $\emptyset$.
\end{ex}

On the other hand, rationalizability is guaranteed under PI choice correspondences.
Thus, they inherit an important property that PI choice functions have.

\begin{thm}\label{thm:rationalize}
Every PI choice correspondence $\C$ is rationalizable.
\end{thm}

Note that this result implies that PI choice correspondences satisfy IRC. Moreover,
it turns out that  PI choice correspondences also satisfy
a stronger IRC condition (see \Cref{lm:restrict}).
Since the PI condition implies substitutability,
a stable matching is guaranteed to exist
when each school has a PI choice correspondence.
The proof of  Theorem \ref{thm:rationalize} is given in the following subsubsection
where some technically important lemmas
that will be used to prove other results
are presented.

\subsubsection{Proof of \Cref{thm:rationalize}}
Throughout this subsubsection, we assume that $\C\colon 2^I\rightrightarrows 2^I$ is a PI choice correspondence.

We first show that any choice $Y\in\C(X)$ is revealed as $C^w(X)=Y$ for some UM weight $w$.
\begin{lm}\label{lm:chw}
For every $X,Y\in 2^I$ with $Y\in\C(X)$, there exists a UM weight $w$ such that $C^w(X)=Y$.
\end{lm}
\begin{proof}
Let $Y=\{i_1,\dots,i_k\}$ and let $I\setminus Y=\{i_{k+1},\dots,i_n\}$.
Define the weight function $w\colon I\to\mathbb{R}$ as follows: 
\begin{align}
w(i_j)&=
\begin{cases}
2^{-j}  & \text{if }j\le k,\\
-2^{-j} & \text{if }j\ge k+1
\end{cases}
\quad(\forall i_j\in I).
\end{align}
Then, $w$ is a UM weight, and we have $C^w(X)=Y$.
\end{proof}
This lemma implies that $\C(X)=\{C^w(X)\mid \text{$w$ is a UM weight}\}$. Thus, a PI choice correspondence $\C$ is representable by a union of PI choice functions.
Nevertheless, the union of PI choice functions is not necessarily a PI choice correspondence (see \Cref{ex:PIunion}).

The next lemma states that a PI choice correspondence $\C$ satisfies a form of idempotent property.
\begin{lm}\label{lm:idemponent}
Let $S \in 2^I$. If $S \in \C(X)$ for some $X \in 2^I$, then $S \in \C(S)$.
\end{lm}
\begin{proof}
By \Cref{lm:chw}, there exists a UM weight $w$ such that $C^w(X)=S$.
As $C^w$ is PI, we have $C^w(S)=C^w(C^w(X))=C^w(X)=S$.
Therefore, $S\in\C(S)$.
\end{proof}

A map $\psi\colon 2^I\to 2^I$ with $\psi(\emptyset)=\emptyset$ is said to be a \emph{closure operator} if the following three properties hold: (extensivity) $X\subseteq\psi(X)$, (idempotence) $\psi(\psi(X))=\psi(X)$, and (monotonicity) $X\subseteq Y$ implies $\psi(X)\subseteq\psi(Y)$.
Define $\tau(X)=\bigcup\{Y\in 2^I\mid \C(X)\cap\C(Y)\ne\emptyset\}$.
We will demonstrate that $\tau$ is a closure operator.\footnote{For a PI choice function $C$, \citet{koshevoy1999choice} proved that the map $\psi(X)=\bigcup\{Y\in 2^I\mid C(X)=C(Y)\}$ is a closure operator that satisfies the anti-exchange property. Our result generalizes this to PI choice correspondences; however, in this case, the anti-exchange property may not hold (e.g., $\C_0$ in \Cref{tab:exChoice}).} 

The following lemma establishes key properties that are essential for demonstrating that $\tau$ is a closure operator.
\begin{lm}\label{lm:tau=}
For $X,S\in 2^I$ with $S\in\C(X)$, we have $\tau(X)=\bigcup\{Y\in 2^I \mid S\in\C(Y)\}$.
Moreover, $S\in\C(\tau(X))$ and $\tau(S)=\tau(X)$.
\end{lm}
\begin{proof}
We write $S^*$ to denote $\bigcup\{Y\in 2^I \mid S\in\C(Y)\}$.
Let $S=\{i_1,\dots,i_k\}$, and let $I\setminus S=\{i_{k+1},\dots,i_n\}$.
Define a weight function $w\colon I\to\mathbb{R}$ as follows:
\begin{align}
w(i_j)=
\begin{cases}
2^{-j}  & \text{if }j\le k,\\
-2^{-j} & \text{if }j\ge k+1
\end{cases}
\quad(\forall i_j\in I).
\end{align}
Then, $w$ is a UM weight, and we have $C^w(X)=S$.

For $Y_1,Y_2\in 2^I$ with $S\in\C(Y_1)$ and $S\in\C(Y_2)$, we have $C^w(Y_1)=S$ and $C^w(Y_2)=S$. Hence, $S=C^w(S)=C^w(C^w(Y_1)\cup C^w(Y_2))=C^w(Y_1\cup Y_2)\in\C(Y_1\cup Y_2)$ since $C^w$ is PI.
This implies that $S=C^w(\bigcup\{Y\in 2^I \mid S\in\C(Y)\})=C^w(S^*)\in\C(S^*)$.

We have $\tau(X)=\bigcup\{Y\in 2^I\mid \C(X)\cap\C(Y)\ne\emptyset\}\supseteq \bigcup\{Y\in 2^I\mid S\in\C(Y)\}=S^*$.
Thus, to obtain $\tau(X)=S^*$, it is sufficient to prove that $\tau(X)\subseteq S^*$.
Suppose to the contrary that $\tau(X)\not\subseteq S^*$.
Then, there exists $T,Z\in 2^I$ such that $T\in\C(X)\cap\C(Z)$ and $Z\not\subseteq S^*$.
Let $\sigma$ be an order such that 
$I\setminus(S\cup T)=\{i_{\sigma(1)},\dots,i_{\sigma(p)}\}$,
$T\setminus S=\{i_{\sigma(p+1)},\dots,i_{\sigma(q)}\}$, and
$S=\{i_{\sigma(q+1)},\dots,i_{\sigma(n)}\}$,
where $0\le p\le q\le n$.
Let $w'\colon I\to\mathbb{R}$ be a UM weight such that
\begin{align}
w'(i_{\sigma(j)})&=
\begin{cases}
-2^{-j} & \text{if }j\le q,\\
2^{-j}  & \text{if }j\ge q+1
\end{cases}
\quad(\forall i_{\sigma(j)}\in I)
\end{align}
Then, $C^{w'}(X)=S$ and $C^{w'}(Z)\subseteq S\cup T~(\subseteq X)$ because $w'(T)>w'(T')$ for any $T'\not\subseteq S\cup T$ (see \Cref{fig:tau=}).
Thus, $C^{w'}(X\cup Z)=C^{w'}(X\cup C^{w'}(Z))=C^{w'}(X)=S$.
This implies $Z\subseteq S^*$, which is a contradiction.
Therefore, $\tau(X)=S^*=\bigcup\{Y\in 2^I \mid S\in\C(Y)\}$.

\begin{figure}[thbp]
\centering
\begin{tikzpicture}[xscale=1, yscale=.9]
\draw[thick, rotate=0] (-0.4,.3) ellipse (2.6cm and 1.8cm);
\node at (-2.5,.5) {$S^*$};

\draw[thick] (0,.2) ellipse (2cm and 1.5cm);
\node at (0,1.4) {$X$};

\draw[thick, fill=bgred, draw=red] (.5,0) ellipse (1.3cm and 1cm);
\node at (1.4,.2) {$T$};

\draw[thick,fill=bggreen,draw=green!50!black] (-.5,0) ellipse (1.3cm and 1cm);
\node at (-1.4,.2) {$S$};
\draw[thick, draw=red!50!black] (.5,0) ellipse (1.3cm and 1cm);

\draw[very thick, myblue] (1,0) ellipse (2cm and 1.2cm);
\node[myblue] at (2.7,.2) {$Z$};

\end{tikzpicture}
\caption{Relations of $X$, $S$, $T$, $Z$, and $S^*$}\label{fig:tau=}
\end{figure}
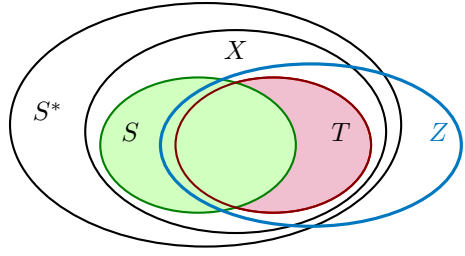

Additionally, $\C(\tau(X))=\C(S^*)\ni C^w(S^*)=S$.
Moreover, as $S\in\C(S)$ by \Cref{lm:idemponent}, it follows that $\tau(X)=\bigcup\{Y\in 2^I \mid S\in\C(Y)\}=\tau(S)$.
\end{proof}

Now we prove that $\tau$ is a closure operator.
\begin{prop}\label{prop:closure}
$\tau$ is a closure operator.
\end{prop}
\begin{proof}
We verify the three properties of a closure operator: extensivity, idempotence, and monotonicity.

\begin{description}[itemsep=.5em,before={\renewcommand{\makelabel}[1]{\textsc{##1}\hfil}}]
\item[(Extensivity)]
For $X\in 2^I$, we have $X\subseteq \bigcup\{Y\in 2^I\mid \C(X)\cap\C(Y)\ne\emptyset\}=\tau(X)$ as $\C(X)\cap\C(X)=\C(X)\ne\emptyset$.

\item[(Idempotence)]
Let $X,S\in 2^I$ with $S\in\C(X)$.
By \Cref{lm:tau=}, we have $S\in\C(\tau(X))$ and $\tau(S)=\tau(X)$.
Applying \Cref{lm:tau=} again to $\tau(X)$ and $S$, we obtain $\tau(S)=\tau(\tau(X))$.
Thus, we concluded that $\tau(X)=\tau(S)=\tau(\tau(X))$.

\item[(Monotonicity)]
For $X\subseteq Y\subseteq I$, let $S=\{i_1,\dots,i_k\}$ be a subset that is in $\C(X)$, and let $I\setminus S=\{i_{k+1},\dots,i_n\}$.
Define a weight function $w\colon I\to\mathbb{R}$ as follows:
\begin{align}
w(i_j)=
\begin{cases}
 2^{-j} & \text{if }j\le k,\\
-2^{-j} & \text{if }j\ge k+1
\end{cases}
\quad(\forall i_j\in I).
\end{align}
Then, $w$ is a UM weight, and we have $C^w(X)=S$.
By \Cref{lm:tau=}, we also have $C^w(\tau(X))=S$.
Moreover, we have 
$C^w(Y\cup \tau(X))=C^w(Y\cup C^w(\tau(X)))=C^w(Y\cup S)=C^w(Y)$.
Hence, 
$Y\cup\tau(X)\subseteq \tau(Y\cup\tau(X))=\tau(C^w(Y\cup \tau(X)))=\tau(C^w(Y))=\tau(Y)$,
where the set inclusion follows from extensivity, and the first and the third equalities are by \Cref{lm:tau=}.
Since $Y\subseteq\tau(Y)$ by extensivity, we conclude $\tau(X)\subseteq \tau(Y)$. \qedhere
\end{description}
\end{proof}

It is known that the inverse image $C^{-1}(X)$ for a PI choice function $C$ forms an interval in $2^I$~\citep{johnson1996algebraic,koshevoy1999choice}. The next lemma generalizes this result to a PI choice correspondence $\C$.
\begin{lm}\label{lm:interval}
For $S,T\in 2^I$ such that $S\in\C(S)$,
we have $S\in\C(T)$ if and only if $S\subseteq T\subseteq \tau(S)$.
\end{lm}
\begin{proof}
Suppose that $S\in\C(T)$. 
By definition, we have $S\subseteq T$.
From \Cref{lm:idemponent}, it follows that $S\in\C(S)$.
Moreover, $T \subseteq \bigcup\{Y \in 2^I \mid \C(S) \cap \C(Y) \neq \emptyset\} = \tau(S)$, since $\C(S) \cap \C(T)~(\ni S)$ is nonempty. Hence, we conclude that $S\subseteq T\subseteq \tau(S)$.

Conversely, suppose that $S\subseteq T\subseteq\tau(S)$.
Let $S=\{i_1,\dots,i_k\}$, and let $I\setminus S=\{i_{k+1},\dots,i_n\}$.
Define a weight function $w\colon I\to\mathbb{R}$ as follows:
\begin{align}
w(i_j)=
\begin{cases}
 2^{-j}  & \text{if }j\le k,\\
-2^{-j} & \text{if }j\ge k+1
\end{cases}
\quad(\forall i_j\in I).
\end{align}
By \Cref{lm:tau=}, we have $C^w(\tau(S))=S$.
Thus, $C^w(T)=C^w(T\cup S)=C^w(T\cup C^w(\tau(S)))=C^w(T\cup\tau(S))=C^w(\tau(S))=S$ by PI of $C^w$.
\end{proof}

The following two lemmas guarantee that PI choice correspondences satisfy a stronger IRC condition.
\begin{lm}\label{lm:taurep}
$\C(X)=\C(\tau(X))\cap 2^X$ holds for any $X\in 2^I$.
\end{lm}
\begin{proof}
First, suppose that $S \in \C(X)$. 
By definition, this implies $S \subseteq X$. 
Furthermore, from \Cref{lm:tau=}, we know that $S \in \C(\tau(X))$. 
Thus, it follows that $S \in \C(\tau(X)) \cap 2^X$.

Conversely, suppose that $S \in \C(\tau(X)) \cap 2^X$. 
This means that $S \in \C(\tau(X))$ and $S \subseteq X$. 
From \Cref{lm:tau=} and \Cref{prop:closure}, we have $\tau(S)=\tau(\tau(X))=\tau(X)\supseteq X$. 
Additionally, by \Cref{lm:idemponent}, we have $S \in \C(S)$. 
Therefore, applying \Cref{lm:interval}, it follows that $S \in \C(X)$.
\end{proof}

\begin{lm}\label{lm:restrict}
Let $X, S \in 2^I$ with $S \in \C(X)$.  
For any $Y \in 2^I$ such that $S \subseteq Y \subseteq X$, it holds that $\C(Y)=\C(X)\cap 2^Y$.\footnote{
This is a strictly stronger condition than the IRC condition.
Indeed, the choice correspondence $\C_{4}$ given in \Cref{tab:exChoice} does not satisfy this condition
for $S = \emptyset$, $Y = \{a\}$, and $X= \{a, b\}$.
}
\end{lm}
\begin{proof}
By \Cref{lm:interval}, we have $S\in\C(Y)$.
By \Cref{lm:tau=}, we have $\tau(Y)=\tau(S)=\tau(X)$.
Hence, by \Cref{lm:taurep}, we obtain
$\C(Y)
=\C(\tau(Y))\cap 2^Y
=\C(\tau(X))\cap 2^Y
=(\C(\tau(X))\cap 2^X)\cap 2^Y
=\C(X)\cap 2^Y$.
\end{proof}

Now we prove \Cref{thm:rationalize}.
\begin{proof}[Proof of \Cref{thm:rationalize}]
Define a utility function $u\colon 2^I\to\mathbb{R}$ as follows:
\begin{align}
u(X)&=
\begin{cases}
|\tau(X)|   & \text{if }X\in\C(X),\\
|\tau(X)|-1 & \text{if }X\notin\C(X)
\end{cases}
\quad(\forall X\in 2^I).
\end{align}
We prove that $u$ rationalizes $\C$.
Let $S\subseteq X$. We show that (i) $u(S)=|\tau(X)|$ if $S\in\C(X)$ and (ii) $u(S)<|\tau(X)|$ if $S\not\in\C(X)$.

(i) If $S\in\C(X)$, then $S\in\C(S)$ by \Cref{lm:idemponent} and $\tau(S)=\tau(X)$ by \Cref{lm:tau=}. Hence, \(u(S)=|\tau(S)|=|\tau(X)|\).

(ii-a) If $S\not\in\C(X)$ and $\tau(S)=\tau(X)$, then $S\not\in\C(S)$ by \Cref{lm:tau=}. Thus,
\(u(S)=|\tau(S)|-1= |\tau(X)|-1\).
(ii-b) If $S\not\in\C(X)$ and $\tau(S)\ne\tau(X)$, then $\tau(S)\subsetneq \tau(X)$ by \Cref{prop:closure}. Therefore, \(u(S)\le |\tau(S)|\le |\tau(X)|-1\).

Thus, $\C(X)=\{S\subseteq X\mid u(S)=|\tau(X)|\}=\argmax\{u(Y)\mid Y\subseteq X\}$ for all $X\in 2^I$.
\end{proof}

\subsection{G-matroid}

We show that a PI choice correspondence
has a nice combinatorial property.
Based on this result, we
present some computational properties 
of the PI choice correspondence.

\begin{thm}\label{thm:gmatroid}
Let $\C$ be a PI choice correspondence.
Then, for every $X\in 2^I$, $\C(X)$ is a g-matroid.
\end{thm}
\begin{proof}
Suppose to the contrary that $\C(X)$ is not a g-matroid.
Then, there exist $S,T\in\C(X)$ and $e\in S\setminus T$ such that
(i) $S-e+e'\not\in\C(X)$ for all $e'\in (T\setminus S)\cup\{\emptyset\}$, or 
(ii) $T+e-e'\not\in\C(X)$ for all $e'\in (T\setminus S)\cup\{\emptyset\}$.
We consider two cases separately.
We remark that, for any $Y$ such that $Y\supseteq S$ or $Y\supseteq T$, it follows that $\C(Y)=\C(X)\cap 2^Y$ by \Cref{lm:restrict}.

\begin{figure}[thbp]
\begin{minipage}{.5\textwidth}
\centering
\begin{tikzpicture}[xscale=1.8,yscale=1.2]
\draw[thick,myblue,fill=bgblue] (-0.2,0) circle ({sqrt(1.04)});
\draw[thick] (-1.7,1.5) rectangle (1.7,-1.3);
\draw[thick] (-0.5,0) circle ({sqrt(5/4)});
\draw[thick] (0.5,0) circle ({sqrt(5/4)});
\node[dot,label=left:$e$] at (-1.3,-0.3) {};
\node[dot,label=above right:$x$] at (0.7,0.3) {};
\node[dot] at (0.72,-0.2) {};
\node at (-1.4,1) {$S$};
\node at ( 1.4,1) {$T$};
\node[myblue] at (0,0.5) {$Z$};
\node at (-1.85,1.2) {$X$};
\end{tikzpicture}
\caption{Case (i)}\label{fig:case-i}
\end{minipage}%
\begin{minipage}{.5\textwidth}
\centering
\begin{tikzpicture}[xscale=1.8,yscale=1.2]
\draw[thick,myblue,fill=bgblue] (0.2,0) circle ({sqrt(1.04)});
\draw[thick] (-1.7,1.5) rectangle (1.7,-1.3);
\draw[thick] (-0.5,0) circle ({sqrt(5/4)});
\draw[thick] (0.5,0) circle ({sqrt(5/4)});
\node[dot,label=left:$e$] at (-0.75,-0.1) {};
\node[dot,label=below right:$i_{q}$] at (1.15,0.65) {};
\node[dot,label=below right:$i_{r+1}$] at (1.25,-0.6) {};
\node at (-1.4,1) {$S$};
\node at ( 1.4,1) {$T$};
\node[myblue] at (0,0.5) {$Z$};
\node at (-1.85,1.2) {$X$};
\end{tikzpicture}
\caption{Case (ii)}\label{fig:case-ii}
\end{minipage}
\end{figure}

\noindent\textbf{Case (i).}
We first consider the case where $S-e+e'\not\in\C(X)$ for all $e'\in (T\setminus S)\cup\{\emptyset\}$.
Let us define the following sets: $S-e=\{i_1,\dots,i_{p-1}\}$,
$e=i_{p}$, and
$I\setminus S=\{i_{p+1},\dots,i_{n}\}$,
where $1\le p< n$.
Next, we define weight functions $w,w'\colon I\to\mathbb{R}$ as follows:
\begin{align}
w(i_j)=
\begin{cases}
2^{-j}  & \text{if }j\le p,\\
-2^{-j} & \text{if }j\ge p+1,
\end{cases}
\quad\text{and}\quad
w'(i_j)=
\begin{cases}
2^{-j}  & \text{if }j\le p-1,\\
-2^{-j} & \text{if }j\ge p.
\end{cases}
\end{align}
Let $Z=C^w(S\cup T-e)$ (see \Cref{fig:case-i}). 
Since $w(i)=w'(i)$ for all $i\in I\setminus\{e\}$, it follows that $Z=C^{w'}(S\cup T-e)$.
Additionally, we have $C^w(X)=S$. 
As $S\cup T-e\subseteq X$ and $C^w$ is substitutable, we have
\begin{align}
S-e
=S\cap(S\cup T-e)
=C^w(X)\cap(S\cup T-e)
\subseteq C^w(S\cup T-e)=Z.
\end{align}
Moreover, we observe that $|Z\setminus S|\ge 2$. If this were not the case (i.e., if $|Z\setminus S|\le 1$), then $Z=S-e+e'$ for some $e'\in (T\setminus S)\cup\{\emptyset\}$, which contradicts the assumption.
Note that, for every $Z'$ such that $S-e\subseteq Z'\subsetneq Z$, we have $Z'\not\in\C(X)$ by the definition of $Z$.
Let $x$ be an element in $Z\setminus S$.

We now demonstrate that $C^{w'}(S+x)=S$.
Note that $S\in \C(S+x)$ by $\C(S+x)=\C(X)\cap 2^{S+x}$.
The set $C^{w'}(S+x)$ must include $S-e$; otherwise we have $w'(S)>w'(C^{w'}(S+x))$, which is a contradiction.
Thus, the possible candidates for $C^{w'}(S+x)$ are $S-e$, $S-e+x$, $S+x$, and $S$.
However, $C^{w'}(S+x)$ cannot be equal to $S+x$ as $w'(S)>w'(S+x)$.
Furthermore, it cannot be equal to either $S-e$ or $S-e+x$, since both sets are not in $\C(S+x)=\C(X)\cap 2^{S+x}$. 
Therefore, the only possibility is that $C^{w'}(S+x)=S$.

Next, we show that $C^{w'}(Z+e)=Z$.
Note that $Z\in \C(Z+e)$ since $\C(Z+e)=\C(X)\cap 2^{Z+e}$.
The set $C^{w'}(Z+e)$ must include $S-e$; otherwise $w'(Z)>w'(C^{w'}(Z+e))$.
Additionally, it cannot include $S$, as $w'(Z)>w'(S')$ for all $S'\supseteq S$.
Hence, $S-e\subseteq C^{w'}(Z+e)\subseteq Z$.
Moreover, no set $Z'$ with $S-e\subseteq Z'\subsetneq Z$ belongs to $\C(Z+e)=\C(X)\cap 2^{Z+e}$. 
Therefore, it follows that $C^{w'}(Z+e)=Z$. 

By combining $C^{w'}(S+x)=S$ and $C^{w'}(Z+e)=Z$, we get
\begin{align}
x\in Z
&=C^{w'}(Z+e)=C^{w'}((S+x)\cup(Z-x))=C^{w'}(C^{w'}(S+x)\cup(Z-x))\\
&=C^{w'}(S\cup(Z-x))=C^{w'}(Z+e-x)\not\ni x.
\end{align}
This is a contradiction.

\addvspace{3mm}
\noindent\textbf{Case (ii).}
Next, we consider the case where $T+e-e'\not\in\C(X)$ for all $e'\in (T\setminus S)\cup\{\emptyset\}$.
Let us define the following sets: $S-e=\{i_1,\dots,i_{p-1}\}$,
$e=i_{p}$, 
$T\setminus S=\{i_{p+1},\dots,i_{q}\}$, and
$I\setminus (S\cup T)=\{i_{q+1},\dots,i_{n}\}$,
where $1\le p<q\le n$.
Define a weight function $w\colon I\to\mathbb{R}$ as follows:
\begin{align}
w(i_j)=2^{-j} \quad (\forall i_j\in I).
\end{align}
Let $Z=C^w(T+e)$ (see \Cref{fig:case-ii}). 
Since $S\in \C(X)\cap 2^{S\cup T}=\C(S\cup T)$, it follows that $C^w(S\cup T)\supseteq S$. By the substitutability of $C^w$, we have
\begin{align}
(T\cap S)+e
=S\cap (T+e)
\subseteq C^w(S\cup T)\cap (T+e)
\subseteq C^w(T+e)=Z.
\end{align}
Moreover, we observe that $|T\setminus Z|\ge 2$. If this were not the case (i.e., if $|T\setminus Z|\le 1$), then $Z=T+e-e'$ for some $e'\in(T\setminus S)\cup\{\emptyset\}$, which contradicts the assumption.
Note that, for every $Z'$ such that $Z\subsetneq Z'\subseteq T+e$, we have $Z'\not\in\C(X)$ by the definition of $Z$.
Let us set $T\setminus Z=\{i_{r+1},\dots,i_q\}$. 
Then, $Z\setminus S=(T\setminus S)\cap Z=\{i_{p+1},\dots,i_r\}$ and $r+1<q$.
Define another UM weight $w'\colon I\to\mathbb{R}$ as follows:
\begin{align}
w'(i_j)=
\begin{cases}
(n+1)\cdot 2^{n-j}   & \text{if }j<p\text{ or }p+1\le j\le r,\\
1+(1/2)^{j} & \text{if }j=p\text{ or }j\ge r+1,
\end{cases}
\quad(\forall i_j\in I).
\end{align}

We now demonstrate that $C^{w'}(T+e)=T$.
Note that $T\in\C(T+e)$ by $\C(T+e)=\C(X)\cap 2^{T+e}$.
The set $C^{w'}(T+e)$ must include $Z-e$; otherwise we have $w'(T)>w'(C^{w'}(T+e))$, which is a contradiction.
Additionally, by the choice of $Z$, we have $Z'\notin\C(T+e)$ for all $Z'$ such that $Z\subsetneq Z'\subseteq T+e$.
Thus, the possible candidates for $C^{w'}(T+e)$ are only $Z$ and $T$.
By the definition of $w'$, we have 
$$w'(T)\ge w'(Z)-w'(e)+w'(i_{q})+w'(i_{r+1})>w'(Z).$$
Therefore, it follows that $C^{w'}(T+e)=T$.

Next, we show that $C^{w'}(Z+i_{q})=Z$.
Note that $Z\in\C(Z+i_q)$ by $\C(Z+i_q)=\C(X)\cap 2^{Z+i_q}$.
The set $C^{w'}(Z+i_{q})$ must include $Z-e$; otherwise we have $w'(Z)>w'(C^{w'}(Z+i_q))$, which is a contradiction.
Hence, the possible candidates for $C^{w'}(Z+i_{q})$ are $Z$, $Z-e$, $Z-e+i_q$, and $Z+i_q$.
It is straightforward to verify that 
$w'(Z)>w'(Z-e)$,
$w'(Z)>w'(Z-e+i_q)$, and
$Z+i_q\not\in\C(T+e)$.
Thus, the only possibility is $C^{w'}(Z+i_{q})=Z$.

Together with $C^{w'}(T+e)=T$ and $C^{w'}(Z+i_{q})=Z$, we obtain
\begin{align}
i_q\in T
&=C^{w'}(T+e)
=C^{w'}((Z+i_q)\cup (T-i_q))
=C^{w'}(C^{w'}(Z+i_q)\cup (T-i_q))\\
&=C^{w'}(Z\cup (T-i_q))
=C^{w'}(T+e-i_q)\not\ni i_q,
\end{align}
which is a contradiction.
\end{proof}

Note that substitutability and IRC are insufficient to obtain \Cref{thm:gmatroid}.
For example, the choice correspondence $\C_{4}$ in \Cref{tab:exChoice} does not induce a g-matroid, as $\C_4(\{a,b\})=\{\emptyset,\{a,b\}\}$, while it satisfies substitutability and IRC.

This theorem implies that for any positive UM weight $w\colon I\to\mathbb{R}_{++}$, the choice $C^w(X)$ is the maximum size in $\C(X)$ by a property of g-matroid. 
\begin{coro}\label{coro:maximal}
Let $\C$ be a PI choice correspondence.
Then, for any positive UM weight $w\colon I\to\mathbb{R}_{++}$, we have $|C^w(X)|=\max\{|Y|\mid Y\in\C(X)\}$.
\end{coro}
\begin{proof}
Let $X^*=C^w(X)$ and suppose that $|X^*|<\max\{|Y|\mid Y\in\C(X)\}$.
As $\C(X)$ is a g-matroid, there is an element $i\in I$ such that $X^*+i\in\C(X)$ by \Cref{prop:gmatroid_nonmaximal}.
This leads to a contradiction as $w(X^*+i)>w(X^*)$.
\end{proof}


In addition, for any UM weight $w\colon I\to\mathbb{R}$, we can construct a membership oracle for $C^w$. 
\begin{coro}\label{coro:reduceCtoF}
For any PI choice correspondence $\C$ and any UM weight $w$, 
we can answer a membership query for $C^w$ in polynomial time by using the membership oracle for $\C$.
\end{coro}
\begin{proof}
Let $X,Y\in 2^I$.
If $Y\notin\C(X)$, then clearly $C^w(X)\ne Y$.
If $Y\in\C(X)$, then $C^w(X)=Y$ (i.e., $w(Y)=\max\{w(X')\mid X'\in\C(X)\}$) if and only if $w(Y)\ge w(Y+u-v)$ for all $u,v\in X\cup\{\emptyset\}$ such that $Y+u-v\in\C(X)$~\cite[Theorem 6.26]{murota2003}.
Since there are at most $O(|X|^2)$ such pairs $(u,v)$, this condition can be verified in $O(|X|^2)$ time.
Consequently, a membership query for $C^w$ can be answered in polynomial time.
\end{proof}

As we can construct a choice oracle from a membership oracle for PI choice functions (see \Cref{sec:computation}), we obtain the following theorem.
\begin{thm}\label{thm:computeC}
Suppose a choice correspondence $\C\colon 2^I\rightrightarrows 2^I$ is accessible via a membership oracle.
Then, for any $X\in 2^I$ and any UM weight $w\colon I\to \mathbb{R}$, we can compute $C^w(X)$ in polynomial time.
\end{thm}

Moreover, if the choice correspondence $\C$ is both PI and LAD, then $C^w(X)$  can be computed more directly and efficiently for every $X\in 2^I$.
\begin{prop}\label{prop:PILADcomp}
Let $\C$ be a choice correspondence that satisfies PI and LAD.
Suppose that we are given $X\in 2^I$ and a UM weight $w\colon I\to\mathbb{R}$. 
Then, we can compute $C^w(X)$ in $O(|X|^2)$ time.
\end{prop}
\begin{proof}
Let $X=\{i_1,i_2,\dots,i_p\}$ and $X_j=\{i_1,i_2,\dots,i_j\}$ for each $j\in\{0,1,\dots,p\}$.
We compute $Y_j$ iteratively as follows.
Set $Y_0=\emptyset$.
For $j=1,2,\dots,p$, define the candidate set
\begin{align}
\mathcal{A}_j\coloneqq \{Y_{j-1},\ Y_{j-1}+i_j\}\cup\{Y_{j-1}-i+i_j\mid i\in Y_{j-1}\}.
\end{align}
Then, choose $Y_j$ such that 
\begin{align}
\{Y_j\}=\argmax\{w(Y)\mid Y\in\mathcal{A}_j\}.
\end{align}
Note that such a unique maximizer exists since $w$ is a UM weight.

We claim that $Y_j=C^w(X_j)$ for every $j$. 
The claim holds for $j=0$ because $C^w(X_0)=C^w(\emptyset)=\emptyset=Y_0$.
Now, assume by induction that $Y_{j-1}=C^w(X_{j-1})$ for some index $j>0$.
By PI of $C^w$, we have $C^w(X_j)=C^w(C^w(X_{j-1})\cup\{i_j\})=C^w(Y_{j-1}+i_j)$.
Moreover, by LAD of $C^w$, we have $|C^w(X_{j})|\ge |C^w(X_{j-1})|=|Y_{j-1}|$.
Thus, $C^w(X_j)$ is either equal to $Y_{j-1}$, $Y_{j-1}+i_j$, or $Y_{j-1}-i+i_j$ for some $i\in Y_{j-1}$.
By our construction, $Y_j$ is chosen from the candidate set $\mathcal{A}_j$ to maximize $w$ among those candidates. Hence, we conclude that $C^w(X_j)=Y_j$.

Therefore, $Y_p=C^w(X_p)=C^w(X)$. 
Note that the iterative process involves $p$ steps.
In each step, the candidate set $\mathcal{A}_j$ contains at most $2+|Y_{j-1}|~(\le p+1)$ candidates. Hence, each iteration requires only $O(p)$ basic operations and membership oracle calls.
Consequently, the overall computational time is at most $O(p^2)=O(|X|^2)$.
\end{proof}

\begin{remark}\label{rem:exp}
Even if a choice correspondence $\C$ can be represented as the union of PI and LAD choice functions, computing $C^w(I)$ for some UM weight $w$ requires an exponential number of queries. Note that, by \Cref{lm:PIunion}, such a choice correspondence also satisfies substitutability and IRC.
To illustrate this, let $I=\{i_1,\dots,i_n\}$, $k=\lfloor n/2\rfloor$, and let $X^*\subseteq I$ be a randomly selected set of size $|X^*|=k+1$.
Additionally, let $\cF=\{X\subseteq I\mid |X|\le k\}\cup\{X^*\}$.
Now, define the choice correspondence $\C$ by $\C(X)=\{X'\subseteq X\mid X'\in\cF\}$.
Note that, by \Cref{prop:general-ub}, $\C$ can be represented as the union of PI and LAD choice functions.
With the UM weight function $w$ specified as $w(i_j)=1+(1/2)^j$ for each $i_j\in I$, the choice $C^w(I)$ is equal to $X^*$.
When querying the membership oracle with a set $X$ of size $k+1$, the oracle reveals only whether $X = X^*$. Since there are exponentially many subsets of size $k+1$, identifying $X^*$ requires an exponential number of queries in expectation.
\end{remark}

\subsection{PI and Ordinal concavity}\label{subsec:ordinal_concave}

It is known that a choice correspondence associated with an ordinally concave function also has the g-matroid property \citep{fujishige2024note}.
Thus, it is natural to examine the relationship between the PI condition and ordinal concavity. The following result shows that a choice correspondence associated with an ordinally concave function is PI, and that size-restricted concavity ensures it is LAD. In particular, this result guarantees that a wide class of choice correspondences arising in real-life applications satisfy both PI and LAD.

\begin{thm}\label{thm:oc-PI}
Any choice correspondence associated with an ordinally concave function satisfies PI. 
Furthermore, any choice correspondence associated with a function that satisfies both ordinal concavity and size-restricted concavity satisfies both PI and LAD.
\end{thm}
\begin{proof}
Let $u\colon 2^I\to\Rbar$ be a utility function.
Fix a UM weight $w\colon I\to\mathbb{R}$, define a utility function $u^w\colon 2^I\to\Rbar$ as follows:
\begin{align}
u^w(X)=u(X)+\delta\cdot w(X) \quad (\forall X\in 2^I),
\end{align}
where $\delta$ is a sufficiently small positive real number such that $u(X)>u(Y)$ implies $u^w(X)>u^w(Y)$. For example, we can select
\begin{align}
\delta=\begin{cases}
1 & \text{if $u$ is a constant function},\\
\frac{\min\{|u(X)-u(Y)|\mid u(X)\ne u(Y)\}}{\max\{1,\, \max\{|w(X)|\mid X\in 2^I\}\}} & \text{otherwise}.
\end{cases}
\end{align}
Let $\C$ be the choice function associated with $u$, and let $C^w$ be its tie-breaking with respect to $w$.
It is straightforward to verify that the choice function $C^w$ is associated with $u^w$.

Suppose that $u$ is an ordinal concave function, i.e., for any $X,X'\in 2^I$ and $i\in X\setminus X'$, there exists $j\in (X'\setminus X)\cup\{\emptyset\}$ such that:
(i) $u(X)<u(X-i+j)$, (ii) $u(X')<u(X'+i-j)$, or (iii) $u(X)=u(X-i+j)$ and $u(X')=u(X'+i-j)$.
We will show that $u^w$ satisfies ordinal concavity.
In case (i), we have $u^w(X)<u^w(X-i+j)$.
Similarly, in case (ii), we have $u^w(X')<u^w(X'+i-j)$.
In case (iii), we have 
\begin{align}
u^w(X)
&=u(X)+\delta\cdot w(X)
=u(X-i+j)+\delta\cdot w(X)\\
&=u^w(X-i+j)-\delta\cdot w(X-i+j)+\delta\cdot w(X)
=u^w(X-i+j)+\delta\cdot (w(i)-w(j))
\end{align}
and 
\begin{align}
u^w(X')
&=u(X')+\delta\cdot w(X')
=u(X'+i-j)+\delta\cdot w(X')\\
&=u^w(X'+i-j)-\delta\cdot w(X'+i-j)+\delta\cdot w(X')
=u^w(X'+i-j)-\delta\cdot (w(i)-w(j)).
\end{align}
As $w$ is a UM weight, $w(i)\ne w(j)$.
Consequently, either $u^w(X)<u^w(X-i+j)$ or $u^w(X)<u^w(X-i+j)$.
Hence, $u^w$ is ordinally concave.
As a choice function associated with an ordinal concavity function is PI, it follows that $C^w$ is PI. Therefore, $\C$ satisfies PI.

Suppose that $u$ additionally satisfies size-restricted concavity, i.e., for any $X,X'\in 2^I$ with $|X|>|X'|$, there exists $i\in X\setminus X'$ such that:
(i) $u(X)<u(X-i)$, (ii) $u(X')<u(X'+i)$, or (iii) $u(X)=u(X-i)$ and $u(X')=u(X'+i)$.
In case (i), we have $u^w(X)<u^w(X-i)$, and in Case (ii), we have $u^w(X')<u^w(X'+i)$. In case (iii), we have $u^w(X)=u^w(X-i)+\delta\cdot w(i)$ and $u^w(X')=u^w(X'+i)-\delta\cdot w(i)$.
Hence, $u^w(X)>u^w(X-i)$ if $w(i)>0$ and $u^w(X')>u^w(X'+i)$ if $w(i)<0$.
Thus, $u^w$ also satisfies size-restricted concavity, and hence $C^w$ is PI and LAD.
\end{proof}

One might expect that the converse of this result holds---that is, every PI choice correspondence is rationalizable by some ordinally concave function.
However, whether this is true remains an open question. 
Even if the answer is negative, we believe that the PI condition is a crucial property for characterizing a class of choice correspondences induced by ordinally concave functions.

M${}^\natural$-concavity is a stronger condition than both ordinal concavity and size-restricted concavity.
Therefore, the above result implies that any choice correspondence associated with an M${}^\natural$-choice function satisfies both PI and LAD.
This fact is particularly useful for applications (see \Cref{sec:application}).
The relationships between these classes of choice correspondences, as well as among the classes defined by substitutability and acceptance, are summarized in \Cref{fig:choice-class}.


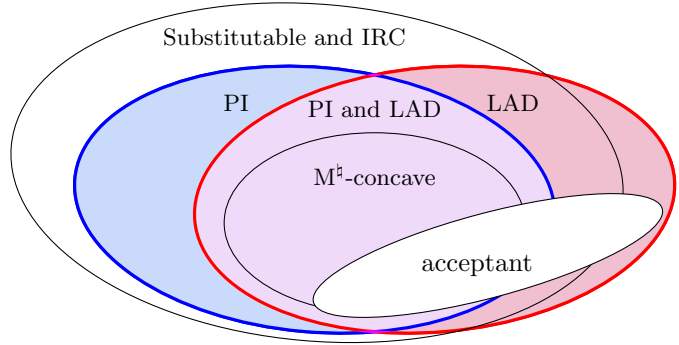
\begin{figure}[ht]
\centering
\scalebox{1}{
\begin{tikzpicture}[xscale=.8, yscale=.8, every node/.style={minimum size=6mm, inner sep=1pt}]

\begin{scope}[blend group=lighten]
\draw[very thick,rotate=-5,draw=blue,fill=bgblue] (-1,0) ellipse (4cm and 2.2cm);
\draw[very thick,rotate=5,draw=red,fill=bgred] (1,0) ellipse (4cm and 2.2cm);
\end{scope}
\begin{scope}[blend group=lighten]
\draw[very thick,rotate=-5,draw=blue] (-1,0) ellipse (4cm and 2.2cm);
\draw[very thick,rotate=5,draw=red] (1,0) ellipse (4cm and 2.2cm);
\end{scope}

\draw[rotate=-5] (-1,0.45) ellipse (5.1cm and 2.8cm);
\node at (-1.5,2.8){\small Substitutable and IRC    };

\node at (-2.3,1.7){\small PI};

\node at (2.3,1.7){\small LAD};
\node at (0,1.6){\small PI and LAD};

\draw (0,-0.3) ellipse (2.5cm and 1.5cm);
\node at (0,.5) {\small M${}^\natural$-concave};


\draw[draw=black,fill=white,fill opacity=.9,rotate=15] (1.6,-1.3) ellipse (3cm and .7cm);
\node at (1.7,-1) {acceptant};

\end{tikzpicture}}
\caption{Classes of choice correspondences}
\label{fig:choice-class}
\end{figure}

\begin{remark}
\citet{FarooqTamura2004} proved that for a utility function $u\colon 2^I\to\Rbar$, the following three conditions are equivalent:
\begin{itemize}
\item[(i)] $u$ satisfies M${}^\natural$-concavity,
\item[(ii)] for any $w\in\mathbb{R}^I$, $\C(X)\coloneqq\argmax\{u(X')+w(X')\mid X'\subseteq X\}$ satisfies (SC$^{1}_{\text{ch}}$), and
\item[(iii)] for any $w\in\mathbb{R}^I$, $\C(X)\coloneqq\argmax\{u(X')+w(X')\mid X'\subseteq X\}$ satisfies (SC$^{2}_{\text{ch}}$).
\end{itemize}
In contrast to their conditions, our property of PI only considers tie-breaking. 
Specifically, we focus on $\C(X)\coloneqq\{u(X')+w(X')\mid X'\subseteq X\}$ for $w\in\mathbb{R}^I$ where $\sum_{i\in I}|w_i|$ is sufficiently small.
\end{remark}

\section{Constrained Efficient Matching}\label{sec:constrainedPE}
In this section, we explore stable and efficient matchings under PI choice correspondences.

\subsection{Matching Model}\label{sec:matching}
A market is a tuple $(I,S,(\succ_i)_{i\in I}, (\C_s)_{s\in S})$, where $I$ is a finite set of students and $S$ is a finite set of schools.
Each student $i\in I$ has a strict preference $\succ_i$ over $S\cup\{\vn\}$, where $\vn$ means being unmatched (or an outside option).
We write $s \succeq_i s'$ if either $s \succ_i s'$ or $s=s'$ holds.

Each school $s \in S$ is endowed with a choice correspondence $\C_s \colon 2^I \rightrightarrows 2^I$. 
The set $\C_s(X)$ represents the most preferred subsets of students in $2^X$ for school $s$, for each $X\in 2^I$. 
For a school $s\in S$ and a UM weight $w$, let $C_s^w$ denote the choice function such that $\{C_s^w(X)\}=\argmax_{Y\in\C_s(X)}w(Y)$ for every $X\in 2^I$.

A \emph{matching} $\mu$ is a subset of $I\times S$ such that each student $i$ appears at most in one pair of $\mu$; that is, $|\mu\cap\{(i,s)\mid s\in S\}|\le 1$ for all $i\in I$.
For each $i\in I$, we write $\mu(i)$ to denote the school to which $i$ is assigned at $\mu$, that is, $\mu(i)=s$ if $(i,s)\in \mu$ and $\mu(i)=\vn$ if $(i,s)\not\in\mu$ for all $s\in S$.
Similarly, for each $s\in S$, we write $\mu(s)$ to denote the set of students assigned to $s$ at $\mu$, that is, $\mu(s)=\{i\in I\mid (i,s)\in\mu\}$.
A matching $\mu$ is called \emph{stable} if it satisfies the following properties:
\begin{itemize}
    \item Individual Rationality: $\mu(i)\succeq_i\varnothing$ for every $i\in I$, and 
    \item No Blocking Coalition: 
    $\mu(s) \in \C_s(\mu(s) \cup X)$ for every $X\subseteq \{i\in I\mid s\succ_i\mu(i)\}$ and $s \in S$.
\end{itemize}
A matching $\mu'$ \emph{Pareto dominates} another matching $\mu$ if $\mu'(i)\succeq_i\mu(i)$ for all $i\in I$ and $\mu'(i)\succ_i\mu(i)$ for some $i\in I$.
A stable matching $\mu$ is \emph{constrained efficient} if it is not Pareto dominated by any other stable matching.

\begin{remark}
Our model can be viewed as a generalization of distributing indivisible goods under constraints studied in \citep{IK2024, IK2024b, suzuki2018efficient, STYYZ2023}.
In these works, a market is defined as a tuple $(I,S,(\succ_i)_{i\in I},(\cF_s)_{s\in S},\mu_0)$, where $\cF_s\subseteq 2^I$ is the family of subsets of students that school $s\in S$ can accept, and $\mu_0$ is the initial matching. 
A matching $\mu$ is called \emph{feasible} if $\mu(i)\succeq_i\vn~(\forall i\in I)$ and $\mu(s)\in\cF_s~(\forall s\in S)$.
A feasible matching $\mu$ is called \emph{Pareto efficient (PE)} if there is no other feasible matching $\mu'$ that Pareto dominates $\mu$.
Additionally, a feasible matching $\mu$ is called \emph{individual rational (IR)} if $\mu(i)\succeq_i\mu_0(i)~(\forall i\in I)$.
We assume that $\mu_0$ is feasible.

For each school $s\in S$, define the choice correspondence $\C_s(X)=\{Y\subseteq X\mid Y\in\cF\}~(\forall X\in 2^I)$. 
With this definition, a matching is feasible if and only if it is stable.
Moreover, a feasible matching that is both PE and IR coincides with a constrained efficient matching that Pareto dominates $\mu_0$, and vice versa.
By \Cref{thm:oc-PI}, $\C_s$ satisfies PI and LAD when $\cF_s\subseteq 2^I$ is a matroid.

Moreover, for the case when every two sets $X',X''\in\cF_s$ satisfy $|X'|=|X''|$ for each $s\in S$, define the choice correspondence $\C'_s(X)=\{Y\subseteq X\mid Y\subseteq Y'\in\cF\}~(\forall X\in 2^I)$.
Then, a feasible matching that is both PE and IR coincides with a constrained efficient matching that Pareto dominates $\mu_0$, and vice versa.
Furthermore, by \Cref{thm:oc-PI}, $\C'_s$ satisfies PI and LAD when $\cF_s\subseteq 2^I$ is a set of matroid bases (i.e., an M-convex set).

Thus, our results in this section are also applicable in these settings.
\end{remark}

As we mentioned in Section \ref{sec:correspondence}, 
a stable matching exists whenever $\C_s$ satisfies substitutability and IRC for all $s \in S$.
Since PI is a stronger condition than these, a stable matching exists if $\C_s$ is PI for all $s \in S$.
In particular, if $C_s^{w_s}$ is PI for all $s\in S$, we can obtain a stable matching in the market $(I,S,(\succ_i)_{i\in I},(\C_s)_{s\in S})$ by applying the \emph{deferred acceptance} (DA) algorithm to the market $(I,S,(\succ_i)_{i\in I},(C_s^{w_s})_{s\in S})$~\citep{roth1984stability,aygun2013matching}.
This is because the outcome of the DA algorithm, $\mu$, satisfies $\mu(i)\succeq_i\vn$ for every $i\in I$ and $\mu(s)=C_s^{w_s}(\mu(s)\cup X)\in \C_s(\mu(s)\cup X)$ for every $X\subseteq \{i\in I\mid s\succ_i\mu(i)\}$ and $s\in S$.
However, the outcome $\mu$ of the DA may not be constrained efficient, as illustrated in the following example.
Thus, tie-breaking may not lead to an efficient stable matching.
This motivates us to explore methods for obtaining a constrained efficient matching from an inefficient stable matching.

\begin{ex}\label{ex:dominated}
Suppose that $I=\{i_1,i_2,i_3,i_4\}$ and $S=\{s_1,s_2,s_3\}$.
The preference $\succ_i$ of each student $i\in I$ is given as follows:
\begin{align}
{\succ_{i_1}}=(s_2\ s_1\ \vn\ s_3),\quad
{\succ_{i_2}}=(s_1\ s_2\ \vn\ s_3),\quad
{\succ_{i_3}}=(s_3\ s_1\ \vn\ s_2),\quad
{\succ_{i_4}}=(s_1\ s_3\ \vn\ s_2).
\end{align}
School $s_1$ has one seat for $\{i_1,i_4\}$ and one seat for $\{i_2,i_3\}$.
Schools $s_2$ and $s_3$ have one seat for $\{i_1,i_2\}$ and one seat for $\{i_3,i_4\}$, respectively.
We assume that each school prefers to fill the seats as much as possible (without prioritizing any specific student).
The resulting choice correspondences $(\C_s)_{s\in S}$ are given as
\begin{align}
\C_{s_1}(X)&=\argmax\{|Y|\mid Y\subseteq X,\ |Y\cap\{i_1,i_4\}|\le 1,\ |Y\cap\{i_2,i_3\}|\le 1\},\\
\C_{s_2}(X)&=\argmax\{|Y|\mid Y\subseteq X,\ |Y\cap\{i_1,i_2\}|\le 1,\ |Y\cap\{i_3,i_4\}|=0\}, \quad (\forall X\in 2^I).\\
\C_{s_3}(X)&=\argmax\{|Y|\mid Y\subseteq X,\ |Y\cap\{i_1,i_2\}|=0,\ |Y\cap\{i_3,i_4\}|\le 1\}
\end{align}
These choice correspondences are PI and LAD as they are derived from weighted matroids.

Consider a matching $\mu=\{(i_1,s_1),(i_2,s_2),(i_3,s_1),(i_4,s_3)\}$.
Then, this matching is stable because $\mu(i)\succeq_i\vn$ for every student $i\in I$, and $\mu(s)\in\C_s(\mu(S)\cup X)$ for every $X\subseteq \{i\in I\mid s\succ_i\mu(i)\}$ and $s\in S$.
However, $\mu$ is not constrained efficient because it is Pareto dominated by another stable matching $\mu'=\{(i_1,s_2),(i_2,s_1),(i_3,s_3),(i_4,s_1)\}$.
Moreover, the matching $\mu=\{(i_1,s_1),(i_2,s_2),(i_3,s_1),(i_4,s_3)\}$ is the outcome of DA with weights $(w_{s_1}(i_1),w_{s_1}(i_2),w_{s_1}(i_3),w_{s_1}(i_4))=(1,4,2,8)$.
\end{ex}

\subsection{Main Result}

Under responsive choice correspondences, constrained efficient matchings are characterized by cycles~\citep{erdil2008s}.
However, in more general settings, this cycle characterization fails \citep{erdil2019efficiency}.
We show that if a choice correspondence satisfies our notions, a similar cycle-based characterization of constrained efficient matchings is restored (Theorem~\ref{thm:main}).
This result has implications for real-life applications that, for example, incorporate diversity requirements.

We introduce two key properties
to characterize constrained efficient stable matchings.
First, we call a stable matching \emph{maximal} if $|\mu(s)|=\max\bigl\{|X|\mid X\in \C_s(\{i\in I\mid s\succeq_i\mu(i)\})\bigr\}$ for every $s\in S$.
We will show that any constrained efficient stable
matching must be maximum (Lemma \ref{lm:sufficiency1}).
Next, we define the notion of a cycle called a \emph{potentially-stable improvement cycle (PSIC)}, which was introduced by \citet{erdil2019efficiency}. 
\begin{df}
A PSIC is a sequence of distinct students $(i_0,i_1,\dots,i_{m-1})$ with $m \ge 2$ such that 
\begin{itemize}
\item $s_\ell\coloneqq\mu(i_\ell)$ for all $\ell\in\{0,1,\dots,m-1\}$,
\item $s_{\ell+1}\succ_{i_\ell}s_{\ell}$ for all $\ell\in\{0,1,\dots,m-1\}$, and
\item $\mu(s_{\ell+1})-i_{\ell+1}+i_\ell \in \C_{s_{\ell+1}}(\{i\in I\mid s_{\ell+1}\succeq_{i}\mu(i)\}-i_{\ell+1})$ for all $\ell\in\{0,1,\dots,m-1\}$,
\end{itemize}
where we treat $i_m=i_0$ and $s_m=s_0$.
\end{df}
We will show that a necessary and sufficient condition for a stable matching to be constrained efficient is that it is maximal and admits no PSIC.

\begin{thm}\label{thm:main}
Suppose that $\C_s$ is PI and LAD for every school $s\in S$.
Then, a stable matching $\mu$ is constrained efficient if and only if it is both maximal and admits no PSIC. 
Moreover, for a given stable matching $\mu$, we can compute a constrained efficient stable matching that Pareto dominates $\mu$ in polynomial time.
\end{thm}

It is worth noting that any mechanism that always produces a constrained efficient stable matching is not strategy-proof, even under the standard responsive choice correspondences~\citep{erdil2008s}.
Consequently, we do not consider strategy-proofness in this work.

The condition of PI takes an important role in Theorem~\ref{thm:main}. 
The following example shows that a constrained efficient matching may admit a PSIC when PI is violated.

\begin{ex}\label{ex:effcient-PSIC}
Consider a market that is almost identical to \Cref{ex:dominated}, but differs only in the choice correspondence of school $s_1$.
In addition to \Cref{ex:dominated}, assume that school $s_1$ cannot accept $i_2$ and $i_4$ at the same time.
The resulting choice correspondence $\C'_{s_1}$ is given as
\begin{align}
\C'_{s_1}(X)
&=\argmax\{|Y|\mid Y\subseteq X,\ |Y\cap\{i_1,i_4\}|\le 1,\ |Y\cap\{i_2,i_3\}|\le 1,\ |Y\cap\{i_2,i_4\}|\le 1\}\ \ \ (\forall X\in 2^I).
\end{align}
For this market, it is not difficult to verify that the matching $\mu=\{(i_1,s_1),(i_2,s_2),(i_3,s_1),(i_4,s_3)\}$ is stable and constrained efficient.
Indeed, $\C'_{s_1}$ satisfies LAD but fails to satisfy PI since the choice function induced by any UM weight $w$ with
$w(i_3) > w(i_4) > w(i_2) > w(i_1)$ does not satisfy PI. Moreover, $(i_1,i_2,i_3,i_4)$ is a PSIC for a constrained efficient matching $\mu=\{(i_1,s_1),(i_2,s_2),(i_3,s_1),(i_4,s_3)\}$.
\end{ex}

Conversely, \citet{erdil2022corrigendum} provided an example where a stable matching that is not constrained efficient but admits no PSIC. In their example, the choice correspondence satisfies substitutability and acceptance but violates PI.\footnote{They provide a condition to obtain the necessity part of Theorem~\ref{thm:main}. We will discuss this in \Cref{sec:admissions}.}

In the rest of this section, we provide the proof of \Cref{thm:main}.
The following lemma characterizes both stable matchings and maximal stable matchings using tie-breaking.
\begin{lm}\label{thm:stable_weight}
Suppose that $\C_s$ is PI for all $s\in S$.
Then, a matching $\mu$ is stable if $\mu(i)\succeq_i \varnothing$ for every $i\in I$ and $\mu(s)=C_s^{w_s}\bigl(\{i\in I \mid s\succeq_i \mu(i)\}\bigr)$ for some UM weight $w_s$, for every $s\in S$.
Moreover, a stable matching $\mu$ is \emph{maximal} if $\mu(s)=C_s^{w_s^+}(\{i\in I\mid s\succeq_i\mu(i)\})$ for some \emph{positive} UM weight $w_s^+\colon I\to\mathbb{R}_{++}$, for every $s\in S$.
\end{lm}
\begin{proof}
Assume that a matching $\mu$ satisfies $\mu(i)\succeq_i \varnothing$ for every $i\in I$ and $\mu(s)=C_s^{w_s}(\{i\in I \mid s\succeq_i \mu(i)\})$ for a UM weight $w_s$, for every $s\in S$.
To prove the stability of $\mu$, it is sufficient to show that $\mu(s)\in\C_s(\mu(s)\cup X)$ for every $X\subseteq \{i\in I\mid s\succ_i\mu(i)\}$ and $s\in S$.
By the PI property of $C_s^{w_s}$, we have
\begin{align}
C_s^{w_s}(\mu(s)\cup X)
&=C_s^{w_s}(C_s^{w_s}(\{i\in I\mid s\succeq_i\mu(i)\})\cup X)\\
&=C_s^{w_s}(\{i\in I\mid s\succeq_i\mu(i)\}\cup X)
=C_s^{w_s}(\{i\in I\mid s\succeq_i\mu(i)\})
=\mu(s).
\end{align}
Thus, we obtain that $\mu(s)=C_s^{w_s}(\mu(s)\cup X)\in \C_s(\mu(s)\cup X)$.

Next, assume that $\mu$ is a stable matching and $\mu(s)=C_s^{w_s^+}(\{i\in I\mid s\succeq_i\mu(i)\})$ for a \emph{positive} UM weight $w_s^+\colon I\to\mathbb{R}_{++}$, for every $s\in S$.
Then, by \Cref{coro:maximal}, we have 
\begin{align}
|\mu(s)|
=|C_s^{w_s^+}(\{i\in I\mid s\succeq_i\mu(i)\})|
=\max\big\{|X|\mid X\in \C_s(\{i\in I\mid s\succeq_i\mu(i)\})\big\},
\end{align}
for every $s\in S$. This means that $\mu$ is maximal.
\end{proof}
It is worth mentioning that, if the choice correspondences are acceptant, every stable matching is maximal.
Unlike the analysis by \citet{erdil2019efficiency} and \citet{erdil2022corrigendum}, we do not assume acceptance; instead, we assume only LAD.
This is important for real-life applications since acceptance is violated while LAD is satisfied under a diversity constraint.

\subsection{Proof of \Cref{thm:main}}
In this subsection, we prove \Cref{thm:main}.

\subsubsection{Sufficiency Part of \Cref{thm:main}}

We now demonstrate the sufficiency direction of Theorem~\ref{thm:main}: if a stable matching is constrained efficient, then it is both maximal and admits no PSIC.  
Unlike the case with responsive choice correspondences studied by \citet{erdil2008s}, not every PSIC necessarily preserves stability. Therefore, it is crucial to select the cycle carefully. The following example illustrates these points.

\begin{ex}\label{ex:shortestPSIC}
Suppose that $I = \{i_1, i_2, i_3, i_4, i_5\}$ and $S=\{s_1, s_2\}$. 
The choice correspondence $\C_{s_1}$ for school $s_1$ is associated with the utility function:
\begin{align}
u_{s_1}(X)&=
\begin{cases}
|X\cap\{i_5\}|+2\cdot \sqrt{|X\cap \{i_1,i_3\}|}+3\cdot|X\cap\{i_2,i_4\}| & \text{if }|X|\le 2,\\
-\infty & \text{if }|X|>2,
\end{cases}
\quad(\forall X\in 2^I).
\end{align}
Similarly, the choice correspondence $\C_{s_2}$ for school $s_2$ is associated with:
\begin{align}
u_{s_2}(X)&=
\begin{cases}
|X|     & \text{if }|X|\le 2,\\
-\infty & \text{if }|X|>2,
\end{cases}
\quad(\forall X\in 2^I).
\end{align}
Since $u_{s_1}$ and $u_{s_2}$ are laminar concave functions, both choice correspondences $\C_{s_1}$ and $\C_{s_2}$ satisfy PI and LAD.
Assume that students $i \in \{i_1, i_3, i_5\}$ have preferences ${\succ_i} = (s_1 \succ s_2 \succ \vn)$, while students $i \in \{i_2, i_4\}$ have preferences ${\succ_i} = (s_2 \succ s_1 \succ \vn)$.
It is straightforward to verify that the matching $\mu = \{(i_1,s_2),(i_2,s_1),(i_3,s_2),(i_4,s_1)\}$ is stable.

In this instance, $(i_1,i_2,i_3,i_4)$ is a PSIC for $\mu$. By applying this PSIC to $\mu$, we obtain another matching $\nu = \{(i_1,s_1),(i_2,s_2),(i_3,s_1),(i_4,s_2)\}$.
However, $\nu$ is not stable since 
$$\C_{s_1}(\{i\in I\mid s_1\succeq \nu(i)\})=\C_{s_1}(\{i_1,i_3,i_5\})=\{\{i_1,i_5\},\{i_3,i_5\}\}\not\ni\{i_1,i_3\}.$$ 
Instead, by applying another PSIC $(i_1,i_2)$, we obtain
$\nu' = \{(i_1,s_1),(i_2,s_2),(i_3,s_2),(i_4,s_1)\}$, which can be verified to be stable.
\end{ex}
The key distinction between the two cycles lies in the presence of a shortcut in the first PSIC. We demonstrate that a PSIC without any shortcuts can preserve stability.\footnote{A PSIC is closely related to a top trading cycle (TTC). Specifically, if each school employs a choice correspondence that returns all feasible subsets of a matroid, then a generalized TTC studied in \cite{IK2024, IK2024b, suzuki2018efficient, STYYZ2023} corresponds to a PSIC.}

In what follows, we first show that if a stable matching is not maximal, then it is not constrained efficient. Second, we demonstrate that if a maximal stable matching admits a PSIC, it cannot be constrained efficient.
\begin{lm}\label{lm:sufficiency1}
Suppose that $\C_s$ satisfies PI and LAD for every school $s\in S$.
If a stable matching $\mu$ is not maximal, then $\mu$ is not constrained efficient.
Moreover, in this case, we can compute another stable matching $\nu$ that Pareto dominates $\mu$ in polynomial time.
\end{lm}
\begin{proof}
Suppose that $\mu$ is a stable matching that is not maximal. 
Then, there exists a school $s \in S$ such that $\mu(s) \not\in C_s^w(\{i \in I \mid s \succeq_i \mu(i)\})$ for any positive UM weight $w \colon I \to \mathbb{R}_{++}$. Fix such a school $s^*$.
By \Cref{thm:gmatroid} and \Cref{prop:gmatroid_nonmaximal}, there exists a student $i^*$ such that $s^*\succ_{i^*}\mu(i^*)$ and $\mu(s^*)+i^*\in\C_{s^*}(\{i\in I\mid s^*\succeq_i\mu(i)\})$.

For each $s\in S\setminus\{s^*\}$, let $w_s$ be a UM weight such that $\mu(s)=C_s^{w_s}(\{i\in I\mid s\succeq_i \mu(i)\})$.
In addition, let $w_{s^*}$ be a UM weight such that $\mu(s^*)+i^*=C_{s^*}^{w_{s^*}}(\{i\in I\mid s^*\succeq_i \mu(i)\})$.
Note that such UM weights can be constructed by setting as in \Cref{lm:chw}.
We construct sequences of matchings $(\mu_0,\mu_1,\dots,\mu_r)$, students $(i_0,i_1,\dots,i_{r-1})$, and schools $(s_0,s_1,\dots,s_{r})$ as follows:
\begin{enumerate}[label=\arabic*.]
\item Initialization:
\begin{itemize}
    \item Set $\mu_0=\mu$, $i_0=i^*$, and $s_0=s^*$.
\end{itemize}
\item Inductive Step ($k=1,2,\ldots$):
\begin{itemize}
    \item Define $\mu_k$ as the matching obtained from $\mu_{k-1}$ by changing the assignment of $i_{k-1}$ from $\mu_{k-1}(i_{k-1})$ to $s_{k-1}$.
    \item Set $s_k=\mu_{k-1}(i_{k-1})$.
    \item If (i) $s_k = \vn$ or (ii) $s_k\in S$ and $C_{s_k}^{w_{s_k}}(\{i\in I\mid s_k\succeq_i \mu_{k}(i)\})=\mu_k(s_k)$, then terminate the process by setting $r=k$. Otherwise, select $i_k$ such that $s_k\succ_{i_k}\mu_k(i_k)$ and $\mu_k(s_k)+i_k=C_{s_k}^{w_{s_k}}(\{i\in I\mid s_k\succeq_i \mu_{k}(i)\})$.
\end{itemize}
\end{enumerate}
We now show that such sequences are always well defined and that the final matching $\mu_r$ is stable and Pareto improves upon the initial matching $\mu$.

We observe that we can select a student $i_k$ at each step $k<r$.
Since $i_0=i^*$, we only consider the case where $k>0$.
Because $C_{s_k}^{w_{s_k}}$ satisfies PI, we have
\begin{align}
C_{s_k}^{w_{s_k}}(\{i\in I\mid s_k\succeq_i \mu_{k}(i)\})
&=C_{s_k}^{w_{s_k}}(\{i\in I\mid s_k\succeq_i \mu_{k-1}(i)\}-i_{k-1})
\supseteq \mu_{k-1}(s_k)-i_{k-1}.
\end{align}
Since $C_{s_k}^{w_{s_k}}$ satisfies LAD, the set $C_{s_k}^{w_{s_k}}(\{i\in I\mid s_k\succeq_i \mu_{k}(i)\})$ is either $\mu_{k-1}(s_k)-i_{k-1}$, or $\mu_{k-1}(s_k)-i_{k-1}+a$ for some student $a\in I$ with $s_k\succ_a\mu_{k}(a)$.
If $C_{s_k}^{w_{s_k}}\big(\{i\in I\mid s_k\succeq_i \mu_{k}(i)\}\big)=\mu_{k-1}(s_k)-i_{k-1}~(=\mu_k(s_k))$, then $k=r$.
Otherwise, if $C_{s_k}^{w_{s_k}}\big(\{i\in I\mid s_k\succeq_i \mu_{k}(i)\}\big)=\mu_{k-1}(s_k)-i_{k-1}+a$ for some student $a\in I$ such that $s_k\succ_a\mu_{k}(a)$, the process continues by setting $i_k=a$.

Next, we prove the following conditions by induction on the step $k$:
\begin{enumerate}[label=(\roman*)]
\item $\mu_k(i)\succeq_i\mu(i)$ for every $i\in I$,
\item $\mu_{k}(s)= C_s^{w_s}(\{i\in I\mid s\succeq_i\mu_{k}(i)\})$ for every $s\in S-s_{k}$, and 
\item if \(k < r\), then \(\mu_k(s_k)+i_k = C_{s_k}^{w_{s_k}}\bigl(\{i \in I \mid s_k \succeq_i \mu_k(i)\}\bigr)\), and\\
if \(k = r\), then $s_k=\vn$, or $s_k\in S$ and \(\mu_k(s_k) = C_{s_k}^{w_{s_k}}\bigl(\{i \in I \mid s_k \succeq_i \mu_k(i)\}\bigr)\).
\end{enumerate}
These conditions hold for the base case ($k=0$) by construction.
For $k>0$, the conditions hold by the choice of $\mu_k$, $i_k$, and $s_k$ as follows.
First, $\mu_k(i)=\mu_{k-1}(i)\succeq_i\mu(i)$ for every $i\in I-i_{k-1}$ and $\mu_k(i_{k-1})=s_{k-1}\succ_{i_{k-1}}\mu_{k-1}(i_{k-1})$.
Second, for every $s\in S\setminus\{s_{k-1},s_k\}$, we have
\begin{align}
\mu_k(s)=\mu_{k-1}(s)= C_s^{w_s}\big(\{i\in I\mid s\succeq_i\mu_{k-1}(i)\}\big)=C_s^{w_s}\big(\{i\in I\mid s\succeq_i\mu_{k}(i)\}\big).
\end{align}
Third, 
\begin{align}
\mu_k(s_{k-1})=\mu_{k-1}(s_{k-1})+i_{k-1}= C_{s_{k-1}}^{w_{s_{k-1}}}\big(\{i\in I\mid s_{k-1}\succeq_{i}\mu_{k-1}(i)\}\big)=C_{s_{k-1}}^{w_{s_{k-1}}}\big(\{i\in I\mid s_{k-1}\succeq_{i}\mu_{k}(i)\}\big).
\end{align}
Finally, $\mu_k(s_k)+i_k= C_{s_k}^{w_{s_k}}\big(\{i\in I\mid s_k\succeq_{i}\mu_k(i)\}\big)$ if $k<r$ and $\mu_r(s_r)= C_{s_r}^{w_{s_r}}\big(\{i\in I\mid s_r\succeq_{i}\mu_k(i)\}\big)$ if $s_r\in S$.
Therefore, the conditions (i)--(iii) hold.

By condition (i), each step in the process results in a Pareto improvement for students. Since there are finitely many students ($|I|$) and schools ($|S|$), the process must terminate after at most $|I| \cdot |S|$ steps.  
By conditions (ii) and (iii), we have $\mu_r(s)=C_s^{w_s}(\{i\in I\mid s\succeq_i\mu_r(i)\})$. Hence, $\mu_r$ is a stable matching.
Therefore, if a stable matching $\mu$ is not maximal, then it cannot be constrained efficient.

Finally, we discuss the computational complexity.
By \Cref{prop:PILADcomp}, we can compute $C_s^{w_s}(X)$ for each $s\in S$ and $X\in 2^I$ in polynomial time.
Since the process has at most $|I|\cdot|S|$ steps and each step involves computations that run in polynomial time, the overall computational complexity is bounded by a polynomial with respect to $|I|$ and $|S|$. Hence, we can find the desired matching $\mu_r$ in polynomial time.
\end{proof}

\begin{lm}\label{lm:sufficiency2}
Suppose that $\C_s$ satisfies PI and LAD for every school $s\in S$.
If a maximal stable matching $\mu$ admits a PSIC, then $\mu$ is not constrained efficient.
Moreover, in this case, we can find another stable matching $\nu$ that Pareto dominates $\mu$ in polynomial time.
\end{lm}
\begin{proof}
Let $(i_0,i_1,\dots,i_{m-1})$ be any PSIC for $\mu$ that does not contain a shortcut.
Define $s_\ell=\mu(i_\ell)$ for $\ell=0,1,\dots,m-1$.

First, we show that there exists a positive UM weight $w\colon I\to\mathbb{R}_{++}$ such that for each $\ell=0,1,\dots,m-1$,
\begin{align}
    \mu(s_{\ell+1})-i_{\ell+1}+i_\ell = C^{w_{s_{\ell+1}}}_{s_{\ell+1}}\big(\{i\in I\mid s_{\ell+1}\succeq_{i}\mu(i)\}-i_{\ell+1}\big),
\end{align}
where we define$i_m=i_0$ and $s_m=s_0$.

For each $s\in S$, let $\sigma\colon I\to\{1,2,\dots,n\}$ be a permutation of $I$ such that
\begin{itemize}
    \item $\sigma(i)<\sigma(j)$ for all $i\in\mu(s)$ and $j\in I\setminus \mu(s)$,
    \item $\sigma(i_k)<\sigma(i_\ell)$ for all $i_k,i_\ell\in\{i_0,i_1,\dots,i_{m-1}\}\cap\nu(s)~(=\nu(s)\setminus\mu(s))$ with $k<\ell$, and
    \item $\sigma(i)<\sigma(j)$ for all $i\in\nu(s)$ and $j\in I\setminus (\mu(s)\cup\nu(s))$.
\end{itemize}
Define positive UM weight $w_s\colon I\to\mathbb{R}_{++}$ by $w_s(i)=2^{n-\sigma(i)}$ for all $i\in I$.

Suppose, for the sake of contradiction, that there exists $\ell\in\{0,1,\dots,m-1\}$ such that 
\begin{align}
\mu(s_{\ell+1}) -i_{\ell+1} + i_\ell \ne C^{w_{s_{\ell+1}}}_{s_{\ell+1}}\big(\{i\in I\mid s_{\ell+1}\succeq_{i}\mu(i)\} - i_{\ell+1}\big).
\end{align}
By the definition of PSIC, we have
\begin{align}
    \mu(s_{\ell+1})-i_{\ell+1}+i_\ell \in \C_{s_{\ell+1}}\big(\{i\in I\mid s_{\ell+1}\succeq_{i}\mu(i)\}-i_{\ell+1}\big).
\end{align}
Since $\mu$ is maximal and $w_{s_{\ell+1}}$ is a positive weight, it follows that
$|\mu(s_{\ell+1})|=|C^{w_{s_{\ell+1}}}_{s_{\ell+1}}(\{i\in I\mid s_{\ell+1}\succeq_{i}\mu(i)\})|$.
By the construction of $w_{s_{\ell+1}}$, we also have
\begin{align}
\mu(s_{\ell+1})-i_{\ell+1} 
\subseteq
C^{w_{s_{\ell+1}}}_{s_{\ell+1}}\big(\{i\in I\mid s_{\ell+1}\succeq_{i}\mu(i)\}-i_{\ell+1}\big).
\end{align}
By LAD, we have $\big|C^{w_{s_{\ell+1}}}_{s_{\ell+1}}\big(\{i\in I\mid s_{\ell+1}\succeq_{i}\mu(i)\}-i_{\ell+1}\big)\big|\le |\mu(s_{\ell+1})|$.
Hence, there exists $i_{k}\in\{i_0,i_1,\dots,i_{\ell-1}\}$ such that 
\begin{align}
\mu(s_{\ell+1})-i_{\ell+1}+i_k 
=
C^{w_{s_{\ell+1}}}_{s_{\ell+1}}\big(\{i\in I\mid s_{\ell+1}\succeq_{i}\mu(i)\}-i_{\ell+1}\big).
\end{align}
This implies that we can construct another PSIC $(i_0,i_1,\dots,i_k,i_{\ell+1},\dots,i_{m-1})$,
which contradicts the assumption that $(i_0,i_1,\dots,i_{m-1})$ does not contain a shortcut.
Therefore, $\mu(s_{\ell+1}) -i_{\ell+1} + i_\ell = C^{w_{s_{\ell+1}}}_{s_{\ell+1}}\big(\{i\in I\mid s_{\ell+1}\succeq_{i}\mu(i)\} - i_{\ell+1}\big)$ for all $\ell\in\{0,1,\dots,m-1\}$.
 
Next, we show that the matching obtained by the cycle $(i_0, i_1, \dots, i_{m-1})$ is stable. Fix any $s\in S$. 
Let $X=\{i_0, i_1, \dots,i_{m-1}\}\cap\mu(s)$ and $Y=\{i_0,i_1,\dots,i_{m-1}\}\cap\nu(s)$.
For each $i_\ell\in Y$, 
we have $i_\ell\in C^{w_s}_s(\{i\in I\mid s\succeq_i \mu(i)\}-i_{\ell+1})$ and $i_{\ell+1}\in X$.
Hence, since $C^{w_s}_s$ satisfies PI, we obtain
\[
\nu(s)=Y\cup(\mu(s)\setminus X) \subseteq C^{w_s}_s(\{i\in I\mid s\succeq_i \mu(i)\}\setminus X).
\]
Moreover, because $C^{w_s}_s$ satisfies LAD, we have
\[
|C^{w_s}_s(\{i\in I\mid s\succeq_i\mu(i)\}\setminus X)|\le|C^{w_s}_s(\{i\in I\mid s\succeq_i\mu(i)\})|=|\mu(s)|.
\]
Furthermore, we have $|\mu(s)|=|\nu(s)|$. 
Together, these facts imply that
\[
\nu(s) = C^{w_s}_s(\{i\in I\mid s\succeq_i\mu(i)\}\setminus X).
\]
Thus, we have $\nu(s) \in \C_s(\{i\in I\mid s\succeq_i\mu(i)\} \setminus X)$, which implies that the matching obtained by the cycle $(i_0, i_1, \dots,i_{m-1})$ is (maximal) stable.

Finally, we observe that a PSIC that does not contain a shortcut can be computed in polynomial time.
We can construct the exchange graph $G=(I,E)$ with directed edges
\begin{align}
E\coloneqq \big\{(i,j)\in I\times I\mid s=\mu(j)\succ_i\mu(i)\text{ and } \mu(s)-j+i\in\C_s(\{i'\in I\mid s\succeq_{i'}\mu(i')\}-j)\big\}
\end{align}
in polynomial time by the membership oracle.
Then, a PSIC is a cycle in this graph $G$ and vice versa.
Thus, our task reduces to finding a minimal cycle in $G$, which can be achieved in polynomial time via breadth-first search.
\end{proof}

\subsubsection{Necessity Part of \Cref{thm:main}}
Next, we prove the necessity part of \Cref{thm:main}. 
\begin{lm}\label{lm:necessity}
If a stable matching $\mu$ is maximal and does not admit any PSIC, then it is constrained efficient.
\end{lm}
\begin{proof}
It suffices to prove that if a maximal stable matching $\mu$ is not constrained efficient, then it admits a PSIC.
Suppose that $\mu$ is Pareto dominated by a constrained efficient stable matching $\nu$.
By \Cref{lm:sufficiency1}, $\nu$ is maximal.
Define $I'=\{i\in I\mid \nu(i)\ne \mu(i)\}$. Note that each $i\in I'$ strictly prefers $\nu$ to $\mu$. 

First, we show that $|\mu(s)|=|\nu(s)|$ for each $s\in S$. 
Since $\nu$ Pareto dominates $\mu$, we have $\{i\in I\mid s\succeq_i\nu(i)\}\subseteq \{i\in I\mid s\succeq_i\mu(i)\}$. 
By \Cref{coro:maximal} and maximality of $\nu$ and $\mu$, we have 
\begin{align}
|\nu(s)|
=|C^{w}_s(\{i\in I\mid s\succeq_i\nu(i)\})|
\le |C^{w}_s(\{i\in I\mid s\succeq_i\nu(i)\})|
= |\mu(s)|,
\end{align}
for any positive UM weight $w\colon I\to\mathbb{R}_{++}$, where the inequality follows from LAD of $C^w_s$.
Suppose to the contrary that $|\nu(s)|< |\mu(s)|$. Then, we have $\sum_{s'\in S}|\nu(s')|<\sum_{s'\in S}|\mu(s')|$. This implies that there exists $i\in I$ such that $\nu(i)=\varnothing$ and $\mu(i)\in S$, contradicting the assumption that $\nu$ Pareto dominates $\mu$.
Hence, $|\mu(s)|=|\nu(s)|$ for all $s\in S$.

Next, we show that there is a cycle on the following directed bipartite graph $(I',J;E)$, where
\begin{align}
J&\coloneqq\big\{(i,\mu(i))\mid i\in I'\big\},\\
E&\coloneqq\big\{\big((i,\mu(i)),\,i\big)\mid i\in I'\big\}\\
&\qquad\cup\pig\{\big(j,\,(i,s)\big)\mid j\not\in\mu(s)\text{ and }\mu(s)-i+j\in \C_{s}\pig(\big(\{i'\in I'\mid s\succ_{i'}\mu(i')\}\cup\mu(s)\big)-i\pig)\pig\}.
\end{align}

To prove the existence of a cycle, it suffices to show that for each $(i,\mu(i))\in J$ there exists $j\in I'$ such that $\big(j,\,(i,\mu(i))\big)\in E$.
Let $s=\mu(i)$. 
Since $\nu$ Pareto dominates $\mu$, we have $\nu(s)\setminus\mu(s)\subseteq \{i'\in I'\mid s\succ_{i'}\mu(i')\}$. Additionally, $\mu(s)\cap \nu(s)\subseteq \mu(s)$ and $i\in \mu(s)\setminus \nu(s)$. 
Hence, it follows that
\begin{align}
\nu(s)\subseteq \big(\{i'\in I'\mid s\succ_{i'}\mu(i')\}\cup\mu(s)\big)-i
\subseteq \{i'\in I\mid s\succeq_{i'}\mu(i')\}. \label{eq:medinclusion}
\end{align}
Since $\mu$ is a maximal stable matching, there is a positive UM weight $w$ such that $\mu(s)= C_s^w(\{i'\in I\mid s\succeq_{i'}\mu(i')\})$. 
By the substitutability of $C^w_s$ and \eqref{eq:medinclusion}, we have
\begin{align}
\mu(s)-i
&=
C^w_s(\{i'\in I\mid s\succeq_{i'}\mu(i')\})\cap\pig(\big(\{i'\in I'\mid s\succ_{i'}\mu(i')\}\cup\mu(s)\big)-i\pig)\\
&\subseteq C^w_s\pig(\big(\{i'\in I'\mid s\succ_{i'}\mu(i')\}\cup\mu(s)\big)-i\pig).
\end{align}
Moreover, by the stability of $\nu$, we have $\nu(s)\in \C_s(\nu(s))$.
Hence, by LAD of $C^w_s$ and \eqref{eq:medinclusion}, we have
\begin{align}
|\nu(s)|
=|C^w_s(\nu(s))|
&\le \big|C^w_s\big(\big(\{i'\in I'\mid s\succ_{i'}\mu(i')\}\cup\mu(s)\big)-i\big)\big|\\
&\le \big|C^w_s\big(\{i'\in I\mid s\succeq_{i'}\mu(i')\}\big)\big|
=|\mu(s)|.
\end{align}
Together with $|\mu(s)|=|\nu(s)|$, we have $|\mu(s)|= \big|C^w_s\big(\big(\{i'\in I'\mid s\succ_{i'}\mu(i')\}\cup\mu(s)\big)-i\big)\big|$.
Thus, there exists $j\in I'\setminus\mu(s)$ such that 
\begin{align}
\mu(s)-i+j=
C^w_s\big(\big(\{i'\in I'\mid s\succ_{i'}\mu(i')\}\cup\mu(s)\big)-i\big)
\in \C_s\big(\big(\{i'\in I'\mid s\succ_{i'}\mu(i')\}\cup\mu(s)\big)-i\big).
\end{align}
Hence, a cycle must exist in the graph.

Finally, we show the existence of a PSIC. Consider any cycle 
$((i_0,s_0),i_0,(i_1,s_1),i_1,\dots,(i_p,s_p),i_p)$ on $(I',J;E)$. 
We will show that $(i_0,i_1,\dots,i_p)$ is a PSIC.
By definition, we have $s_\ell=\mu(i_\ell)\ne s_{\ell+1}$, and for every $\ell\in\{0,1,\dots,p\}$
\begin{align}
\mu(s_{\ell+1})-i_{\ell+1}+i_{\ell}\in \C_{s_{\ell+1}}\Big(\big(\{i'\in I'\mid s_{\ell+1}\succ_{i'}\mu(i')\}\cup\mu(s_{\ell+1})\big)-i_{\ell+1}\Big).
\label{eq:midin}
\end{align}
This implies that $s_{\ell+1}\succ_{i_\ell}s_\ell$ for every $\ell\in\{0,1,\dots,p\}$.
Thus, to prove that $(i_0,i_1,\dots,i_p)$ is a PSIC, it suffices to show that for all $\ell\in\{0,1,\dots,p\}$,
\begin{align}
\mu(s_{\ell+1})-i_{\ell+1}+i_{\ell}\in \C_{s_{\ell+1}}\big(\{i\in I\mid s_{\ell+1}\succeq_i \mu(i)\}-i_{\ell+1}\big).
\end{align}
Assume for contradiction that there exists some index $\ell$ such that
\begin{align}
\mu(s_{\ell+1})-i_{\ell+1}+i_{\ell}\notin \C_{s_{\ell+1}}\big(\{i\in I\mid s_{\ell+1}\succeq_i \mu(i)\}-i_{\ell+1}\big).
\label{eq:midout}
\end{align}
Fix such an index $\ell$.

Similar to \eqref{eq:medinclusion}, we have
\begin{align}\label{eq:medinclusion2}
\begin{split}
\nu(s_{\ell+1})
&\subseteq \big(\{i'\in I'\mid s_{\ell+1}\succ_{i'}\mu(i')\}\cup\mu(s_{\ell+1})\big)-i_{\ell+1}\\
&\subseteq \{i\in I\mid s_{\ell+1}\succeq_{i}\mu(i)\}-i_{\ell+1}
\subseteq \{i\in I\mid s_{\ell+1}\succeq_{i}\mu(i)\}. 
\end{split}
\end{align}
By \eqref{eq:midin}, \eqref{eq:midout}, and \Cref{lm:interval}, we have 
\begin{align}
\C_{s_{\ell+1}}\Big(\big(\{i'\in I'\mid s\succ_{i'}\mu(i')\}\cup\mu(s_{\ell+1})\big)-i_{\ell+1}\Big)
\cap\C_{s_{\ell+1}}\big(\{i\in I\mid s_{\ell+1}\succeq_i \mu(i)\}-i_{\ell+1}\big)
=\emptyset. \label{eq:nonintersect}
\end{align}

Let $\sigma\colon\{1,2,\dots,n\}\to I$ be a bijection such that 
$\mu(s_{\ell+1})=\{\sigma(1),\dots,\sigma(p)\}$,
$\nu(s_{\ell+1})\setminus\mu(s_{\ell+1})=\{\sigma(p+1),\dots,\sigma(q)\}$, and 
$I\setminus(\nu(s_\ell)\cup\mu(s_\ell))=\{\sigma(q+1),\dots,\sigma(n)\}$.
Let $w$ be a positive UM weight defined by $w(\sigma(t))=2^{n-t}$ for each $t\in\{1,2,\dots,n\}$.
Then, $\mu(s_{\ell+1})=C_{s_{\ell+1}}^w(\{i\in I\mid s_{\ell+1}\succeq_i\mu(i)\})$ and $\nu(s_{\ell+1})=C_{s_{\ell+1}}^w(\nu(s_{\ell+1}))$.
By the substitutability of $C_{s_{\ell+1}}^w$ and \eqref{eq:medinclusion2}, we have
\begin{align}
\mu(s_{\ell+1})-i_{\ell+1}
&=C_{s_{\ell+1}}^w\big(\{i\in I\mid s_{\ell+1}\succeq_i \mu(i)\}\big)\cap\big(\{i\in I\mid s_{\ell+1}\succeq_i \mu(i)\}-i_{\ell+1}\big)\\
&\subseteq C_{s_{\ell+1}}^w\big(\{i\in I\mid s_{\ell+1}\succeq_i \mu(i)\}-i_{\ell+1}\big).
\label{eq:mu-i}
\end{align}
By \eqref{eq:medinclusion2} and LAD of $C_{s_\ell}^w$, we obtain
\begin{align}
|\nu(s_{\ell+1})|
=|C_{s_{\ell+1}}^w(\nu(s_{\ell+1}))|
&\le |C_{s_{\ell+1}}^w(\{i\in I\mid s_{\ell+1}\succeq_i \mu(i)\}-i_{\ell+1})|\\
&\le |C_{s_{\ell+1}}^w(\{i\in I\mid s_{\ell+1}\succeq_i \mu(i)\})|
=|\mu(s_{\ell+1})|.
\label{eq:nulemu}
\end{align}
By combining \eqref{eq:mu-i}, \eqref{eq:nulemu}, and $|\mu(s_{\ell+1})|=|\nu(s_{\ell+1})|$, there exists a student $j\in \{i\in I\mid s_{\ell+1}\succ_i\mu(i)\}$ such that
\begin{align}
\mu(s_{\ell+1})-i_{\ell+1}+j
= C_{s_{\ell+1}}^w(\{i\in I\mid s_{\ell+1}\succeq_i \mu(i)\}-i_{\ell+1})
\in \C_{s_{\ell+1}}(\{i\in I\mid s_{\ell+1}\succeq_i \mu(i)\}-i_{\ell+1}).
\label{eq:mu-j}
\end{align}
In what follows, we consider two cases depending on whether (a) $j\in I'$ and (b) $j\notin I'$.

\addvspace{3mm}
\noindent\textbf{Case (a)}: Suppose that $j\in I'$.
Then, by \eqref{eq:medinclusion2}, \eqref{eq:mu-j}, and PI of $C_{s_{\ell+1}}^w$, we have
\begin{align}
\mu(s_{\ell+1})-i_{\ell+1}+j
&= C^w_{s_{\ell+1}}\big(\big(\{i'\in I'\mid s_{\ell+1}\succ_{i'}\mu(i')\}\cup\mu(s_{\ell+1})\big)-i_{\ell+1}\big)\\
&\in \C_{s_{\ell+1}}\big(\big(\{i'\in I'\mid s_{\ell+1}\succ_{i'}\mu(i')\}\cup\mu(s_{\ell+1})\big)-i_{\ell+1}\big).
\end{align}
Consequently, \Cref{lm:interval} implies that 
\begin{align}
\C_{s_{\ell+1}}\Big(\big(\{i'\in I'\mid s\succ_{i'}\mu(i')\}\cup\mu(s_{\ell+1})\big)-i_{\ell+1}\Big)
\subseteq \C_{s_{\ell+1}}\big(\{i\in I\mid s_{\ell+1}\succeq_i \mu(i)\}-i_{\ell+1}\big).
\end{align}
This contradicts \eqref{eq:nonintersect}.

\addvspace{3mm}
\noindent\textbf{Case (b)}: Suppose that $j\notin I'$. 
Then, $j\in I\setminus(\mu(s_{\ell+1})\cup\nu(s_{\ell+1}))$ and 
$s_{\ell+1}\succ_j \mu(j)=\nu(j)$.
Since $\nu$ is stable, we have $\nu(s_{\ell+1})\in \C_{s_{\ell+1}}(\{i\in I\mid s_{\ell+1}\succeq_i \nu(i)\})$.
Moreover, 
\begin{align}
\mu(s_{\ell+1})\setminus\nu(s_{\ell+1})
\subseteq \{i\in I\mid \nu(i)\succ_i s_{\ell+1}\}
= I\setminus \{i\in I\mid s_{\ell+1}\succeq_i\nu(i)\}
\end{align}
since $\nu$ Pareto dominates $\mu$. Thus, by the construction of $w$, we have $\nu(s_{\ell+1})= C^{w}_{s_{\ell+1}}(\{i\in I\mid s_{\ell+1}\succeq_i \nu(i)\})$. Hence, by \eqref{eq:medinclusion2}, \eqref{eq:mu-j}, and the substitutability of $C^{w}_{s_{\ell+1}}$, we have
\begin{align}
j
&\in (\mu(s_{\ell+1})-i_{\ell+1}+j)\cap 
\{i\in I\mid s_{\ell+1}\succeq_{i}\nu(i)\}\\
&=C_{s_{\ell+1}}^w\big(\{i\in I\mid s_{\ell+1}\succeq_{i}\mu(i)\}-i_{\ell+1}\big)\cap 
\{i\in I\mid s_{\ell+1}\succeq_{i}\nu(i)\}\\
&\subseteq C_{s_{\ell+1}}^w(\{i\in I\mid s_{\ell+1}\succeq_{i}\nu(i)\})=\nu(s_{\ell+1})\not\ni j,
\end{align}
which is a contradiction.
\end{proof}

Finally, we prove \Cref{thm:main}.
\begin{proof}[Proof of \Cref{thm:main}]
By Lemmas~\ref{lm:sufficiency1} and \ref{lm:sufficiency2}, every constrained efficient stable matching is maximal and admits no PSIC. Conversely, by \Cref{lm:necessity}, any stable matching that is maximal and does not admit a PSIC is constrained efficient.

If a given stable matching is not maximal, we can compute a Pareto-improving stable matching in polynomial time, as demonstrated in \Cref{lm:sufficiency1}. 
Additionally, if a maximal stable matching admits a PSIC, a Pareto-improving stable matching can be obtained in polynomial time, as shown in \Cref{lm:sufficiency2}. 
Since the number of possible Pareto improvements is at most $|I| \cdot |S|$, the overall computational time is bounded by a polynomial.
\end{proof}

\begin{remark}\label{rem:improve-hard}
\citet{IK2024b} proved that checking whether the initial matching $\mu_0$ is PE for a market $(I,S,(\succ_i)_{i\in I},(\cF_s)_{s\in S},\mu_0)$ is coNP-hard, even when the constraints $\cF_s$ are budget constraints (i.e., constants that can be represented in the form $\{X\subseteq I\mid \sum_{i\in X}a_i\le b\}$).
Hence, checking the constrained efficiency of a given stable matching is a difficult task when the choice correspondences may not satisfy PI and LAD.
Note that for such a constraint, the associated choice correspondence $\C_s(X)=\{Y\subseteq X\mid Y\in\cF_s\}$ satisfies substitutability and IRC (see \Cref{sec:sub_PI}).
\end{remark}

\section{Applications}\label{sec:application}
In this section, we explore several examples of practical choice correspondences and demonstrate their diverse applications, particularly in the context of matching theory.
Notably, recall that every choice correspondence rationalized by an M${}^\natural$-concave function satisfies PI and LAD.

\subsection*{Responsive Choice Correspondences}
\citet{abdulkadirouglu2003school} studied school choice problems using responsive choice functions. In practice, however, a school’s priority ranking often includes ties. For example, if priority is determined by test scores, applicants with the same test score are tied. Moreover, in situations such as the Boston public school choice system---where only neighborhood and sibling priorities are considered---many ties can occur. In such cases, each school should have a responsive choice correspondence. 

Each school $s$ has a capacity $q_s\in\mathbb{Z}_{+}$ and a weak order $\succeq_s$ over $I$. 
For every school $s\in S$, a responsive choice correspondence $\C_s$ is rationalized by utility function $u_s$ that is defined as follows: there exists a valuation $v_s\colon I\to\mathbb{R}_{++}$ satisfying $v_s(i)\ge v_s(j)$ if and only if $i\succeq_s j$ for all $i,j\in I$ such that for each $X\in 2^{I}$,
\begin{align}
u_s(X)=
\begin{cases}
\sum_{i\in X}v_s(i) & \text{if }|X|\le q_s,\\
-\infty             & \text{if }|X|>q_s.
\end{cases}
\end{align}
Since $u_s$ is derived from a weighted matroid with a uniform matroid of rank $q_s$, the utility function $u_s$ induces a choice correspondence that is PI and LAD. 

It is worth mentioning that the resulting choice correspondence remains the same for any other valuation $v'\colon I\to\mathbb{R}_{++}$ satisfying $v'(i)\ge v'(j)$ if and only if $i\succeq_s j$ for all $i,j\in I$ due to a property of matroids.

\citet{erdil2008s} provided a cycle‐based characterization of constrained efficient matching under responsive choice correspondences. We obtain this result as a corollary of our Theorem~\ref{thm:main} because any responsive choice correspondence satisfies PI and LAD.

\subsection*{Controlled School Choice}

A school district may require specific diversity in the student body at each school. \citet{abdulkadirouglu2003school} formalized this requirement by imposing type-specific quotas for each school. \citet{hafalir2013effective} proposed an affirmative action policy based on minority reserves. \citet{ehlers2014school} incorporated these ideas and introduced type-specific (soft) quotas and reserves.

\citet[online appendix]{kojima2018designing} showed that these choice functions can be rationalized by M$^{\natural}$-concave functions. Moreover, we observe that this result applies to weak priorities (i.e., choice correspondence). 

Suppose that $(I_t)_{t\in T}$ is a partition of students with types $T$, i.e., $\bigcup_{t\in T}I_t=I$ and $I_t\cap I_{t'}=\emptyset$ for all $t,t'\in T$ with $t\ne t'$.
We write $t(i)$ to denote the type of $i\in I$. Thus, $i\in I_{t(i)}$.
Each school $s$ has a capacity $q_s\in\mathbb{Z}_{+}$ and soft minimum and maximum bounds for each type $t$, denoted by $\underline{q}_{s,t}$ and $\overline{q}_{s,t}$, respectively.
We assume $\sum_{t\in T}\underline{q}_{s,t}\le q_s$ holds. 
In addition, each school $s$ has a weak order $\succeq_s$ over $I$.
Let $v_s\colon I\to\mathbb{R}_{++}$ be a valuation satisfying $v_s(i)\ge v_s(j)$ if and only if $i\succeq_s j$ for all $i,j\in I$.
Then, the choice correspondence $\C_s$ is rationalizable by
$$
u_s(X)=
\begin{cases}
\sum_{i\in X} \left(1+\epsilon^2 v_s(i)\right)+\epsilon\sum_{t\in T}\Big(\min\left\{|X_t|,\,\underline{q}_{s,t}\right\}+\min\left\{|X_t|,\,\overline{q}_{s,t}\right\}\Big) & \text{if }|X|\le q_s,\\
-\infty             & \text{if }|X|>q_s,
\end{cases}
$$
where $\epsilon$ is a sufficiently small positive real number.
Thus, $\C_s$ is rationalizable by a laminar concave function for $\cL=\{\{i\}\mid i\in I\}\cup \{I_t\mid t\in T\}\cup \{I\}$, and hence, $\C_s$ is PI and LAD.

\subsection*{Evenly Distributed and Constrained Responsive Choice Correspondences}
\citet{erdil2019efficiency} studied how symmetric treatment of types can be implemented with type-specific reserves.
Suppose that $(I_t)_{t\in T}$ is a partition of students with types $T$, i.e., $\bigcup_{t\in T}I_t=I$ and $I_t\cap I_{t'}=\emptyset$ for all $t,t'\in T$ with $t\ne t'$.
We write $t(i)$ to denote the type of $i\in I$. Thus, $i\in I_{t(i)}$.
Each school $s$ has a capacity $q_s\in\mathbb{Z}_{+}$ and type-specific reserves $r_s\in \mathbb{Z}_+^{T}$ with $\sum_{t\in T}r_{s,t}\le q_s$. In addition, each school $s$ has a weak order $\succeq_s$ over $I$.
For each $X\in 2^I$, surplus seats ($q_s-\sum_{t\in T}\min\{|X_t|,r_{s,t}\}$) are distributed evenly among types.
They introduced evenly distributed and constrained responsive (EDCR) choice correspondences $\C_s$. This choice correspondence satisfies acceptance. For each $X\in 2^I$ with $|X|>q_S$, this choice proceeds in two stages. 
\begin{enumerate}
\item It selects subsets of students $X'\subseteq X$ of size $q_s$ that minimizes $\sum_{t\in T}(r_{s,t}-|X'_t|)^2$. 
\item Among the subsets chosen in the first stage, it selects best subsets with respect to a weak priority $\succeq_s$.
\end{enumerate}
For each $s\in S$, let $v_s\colon I\to\mathbb{R}_{++}$ be a positive weight such that $v_s(i)\ge v_s(j)$ if and only if $i\succeq_s j$ for all $i,j\in I$.
Then, the choice correspondence $\C_s$ is rationalizable by
$$
u_s(X)=
\begin{cases}
|X|-\epsilon\sum_{t\in T}(r_{s,t}-|X_t|)^2+\epsilon^2\sum_{i\in X} v_s(i)+\epsilon\sum_{t\in T}r_{s,t}^2 & \text{if }|X|\le q_s,\\
-\infty             & \text{if }|X|>q_s,
\end{cases}
$$
where $\epsilon$ is a sufficiently small positive real number.
By a simple calculation, we obtain
\begin{align}
|X|-\epsilon\sum_{t\in T}(r_{s,t}-|X_t|)^2+\epsilon^2\sum_{i\in X} v_s(i)
+\epsilon\sum_{t\in T}r_{s,t}^2&=-\epsilon\sum_{t\in T}|X_t|^2+\sum_{i\in X} (1+2\epsilon r_{s,t(i)}+\epsilon^2 v_s(i)).
\end{align}
Thus, $u_s$ is a laminar concave function for $\cL=\{\{i\}\mid i\in I\}\cup \{I_t\mid t\in T\}\cup \{I\}$.
Hence, $\C_s$ is PI and LAD.



\subsection*{Overlapping Reserves}
In practice, each student can have multiple types. In practice, each student can have multiple types. One example is affirmative action policies that account for both racial and income minorities. There are several ways to count a student with multiple types toward a reserved seat~\citep{kurata2017controlled}. Let $T$ be the set of types and $I_t\subseteq I$ be the set of students with type $t\in T$.
Each school $s$ has a capacity $q_s\in\mathbb{Z}_{+}$ and type-specific reserves $r_s\in \mathbb{Z}_+^{T}$ with $\sum_{t\in T}r_{s,t}\le q_s$\footnote{This condition could be removed without affecting the construction of the PI choice correspondence. However, it is included to ensure that the concept of ``reserve'' is guaranteed.}. In addition, each school $s$ has a weak order $\succeq_s$ over $I$.
We focus on \emph{one-to-one counting}, where a student counts toward a reserved seat as only one of her types. 
This model includes important real-life applications, such as affirmative action in India~\citep{sonmez2022affirmative} and Brazil~\citep{aygun2021college}.

In this setting, \citet{sonmez2022affirmative} proposed a \emph{meritorious horizontal choice function} for cases where each school has strict priority $\succ_s$ over $I$. This function is designed to maximize the reserve utilization. Formally, the sets of students that can be assigned to reserved seats for $s$ are represented by
$$\cF_s=\{X\subseteq I\mid \exists \pi\colon X\to T\text{ such that }\pi^{-1}(t)\subseteq I_t~(\forall t\in T)\text{ and } |\pi^{-1}(t)|\le r_{s,t}~(\forall t\in T)\}.$$
It is known that $\cF_s$ is a (transversal) matroid.
A meritorious horizontal choice function consists of two stages:
\begin{enumerate}
\item In the first stage, the best subset of students in $\cF_s$ is selected based on $\succ_s$.
\item In the second stage, the remaining seats are assigned to the best remaining students based on $\succ_s$.
\end{enumerate}

In what follows, we observe that a meritorious horizontal choice function is rationalizable by an M${}^\natural$-concave function.
Construct a bipartite graph $G=(I,J;E)$ with weight $w_e\in\mathbb{R}$ for each $e\in E$, where 
\begin{itemize}
\item $J\coloneqq H\cup\bigcup_{t\in T} P_t$ with $P_t\coloneqq\{p_{t,1},\dots,p_{t,r_{s,t}}\}$ ($\forall t\in T$) and $H\coloneqq\{h_1,\dots,h_{q_s}\}$,
\item $E\coloneqq \{(i,h)\mid i\in I,\ h\in H\}\cup\bigcup_{t\in T}\{(i,p)\mid i\in I_t,\ p\in P_t\}$.
\end{itemize}
Additionally, let $\cJ=\{J'\subseteq J\mid |J'|\le q_s\}$ be the uniform matroid of rank $q_s$ on $J$.
For $M \subseteq E$, we denote by $\partial M$ the set of the vertices incident to some edge in $M$, and call $M$ a \emph{matching} if $|I \cap \partial M| = |M| = |J \cap \partial M|$.
For $X\subseteq I$, we write $u_s(X)$ to denote the maximum weight of a matching $M$ such that the end-vertices in $I$ are equal to $X$ and the end-vertices in $J$ form an independent set, i.e.,
\begin{align}
u_s(X)=\max\left\{\sum_{e\in M}w_e\mid 
\text{$M\subseteq E$ is a matching},\ 
I\cap \partial M=X,\ 
J\cap \partial M\in\cJ
\right\},
\end{align}
where $u_s(X)=-\infty$ if no such $M$ exists for $X$.
Such $u_s$ is called an \emph{independent assignment valuation}. It is known that an independent assignment is an M${}^\natural$-concave function~\cite[Section 3.6]{murota2016}). 
Thus, $u_s$ induces a choice correspondence that is PI and LAD.

Let $v_s\colon I\to\mathbb{R}_{++}$ be a positive weight such that $v_s(i)> v_s(j)$ if and only if $i\succ_s j$ for all $i,j\in I$.
Then, the utility function $u_s$ induces the meritorious horizontal choice function by setting $w_{(i,p)}=v_s(i)+M$ for $(i,p)\in \bigcup_{t\in T}(I_t\times P_t)$ and $w_{(i,h)}=v_s(i)$ for $(i,h)\in I\times H$, where $M$ is a sufficiently large positive real number (e.g., $M=1+\sum_{i\in I}v_s(i)$).
Moreover, by adjusting the weight settings, it is also possible to represent choice correspondences in cases where the priority is weakly ordered or where the priority changes depending on which type is adopted.

\section{Conclusion}\label{sec:conclusion}
In this paper, we introduced PI choice correspondences and examined their key properties, including rationalizability and the g-matroid structure. Building on these properties, we developed a characterization of constrained efficient stable matching using a PSIC. Additionally, we highlighted the broad applicability of PI choice correspondences by leveraging M${}^\natural$-concave functions within the framework of discrete convex analysis.

Some choice correspondences lie outside the scope of our framework. 
For instance, \citet{che2019weak} examines those motivated by multidivisional organizations and regional caps; however, these correspondences are not rationalizable. Furthermore, \citet{erdil2019efficiency} introduced a choice correspondence---termed \emph{admissions by a committee}---and showed that constrained efficient matchings can still be characterized by cycles. Nonetheless, as we demonstrate in \Cref{ex:admission} in the appendix, this type of correspondence may not be rationalizable. 
Developing a more general theory that incorporates such choice correspondences remains a promising direction for future research.


\section*{Acknowledgments}
This work was partially supported by JSPS KAKENHI Grant Numbers JP21H04979, JP23K01312, and JP23K12443, JST PRESTO Grant Number JPMJPR2122, JST ERATO Grant Number JPMJER2301, and Value Exchange Engineering, a joint research project between Mercari, Inc.\ and the RIISE.

\bibliographystyle{abbrvnat}
\bibliography{pe}

\appendix

\section{Characterization of Rationalizability}\label{sec:rationalize}
In this section, we present a characterization of rationalizability for choice correspondences.
Our result generalizes the characterization of rationalizability for choice functions provided by \citet{yang2020rationalizable} to choice correspondences.
Yang demonstrated that a choice function is rationalizable if and only if it satisfies the \emph{strong axiom of revealed preference (SARP)}~\citep{aygun2013matching}.
Unlike choice functions, where a single unique choice is considered, choice correspondences allow for multiple selections from the same set, requiring us to account for situations where different sets are assigned the same value. This additional consideration is crucial when analyzing rationalizability.


Define 
\begin{align}
\Gamma &\coloneqq \{X \in 2^I \mid X \in \C(Y) \text{ for some } Y \subseteq I\},\\
P &\coloneqq \big\{(X, Y) \in \Gamma \times \Gamma \mid \{X, Y\} \subseteq \C(Z) \text{ for some } Z \supseteq X \cup Y\big\}.
\end{align}
If a utility function $u$ induces $\C$, then we have $u(X)=u(Y)$ for all $(X,Y)\in P$.
Note that $P$ is a symmetric (i.e., if $(X, Y) \in P$, then $(Y, X) \in P$) and reflexive (i.e., $(X, X) \in P$ for all $X \in \Gamma$) binary relation.
Let $\sim$ denote the transitive closure of $P$. Then, $\sim$ is an equivalence relation. Define 
\begin{align}
\Gamma' = \{[X] \mid X \in \Gamma\},
\end{align}
where $[X] = \{Y \in \Gamma \mid X \sim Y\}$ is the equivalence class of $X$ under $\sim$.
If a utility function $u$ induces $\C$, then it holds that $u(X)=u(Y)$ for all $X,Y\in\Gamma$ such that $X\sim Y$.
Next, define 
\begin{align}
Q \coloneqq \big\{([X], [Y]) \mid X \in \C(Z) \text{ and } Y \notin \C(Z) 
\text{ for some } X, Y \in \Gamma 
\text{ and } Z \subseteq I 
\text{ with } Z \supseteq X \cup Y\big\}.
\end{align}
Finally, let $\succ$ denote the transitive closure of $Q$.
Intuitively, if a utility function $u$ induces $\C$, then we have $u(X)>u(Y)$ for all $X,Y\in\Gamma$ such that $X\succ Y$.

We characterize rationalizability by using a strict partial order.
A homogeneous relation $\succ$ is strict partial order if it satisfies (i) transitivity, (ii) irreflexivity (i.e., $X\not\succ X$), and (iii) asymmetry (i.e., $X\succ Y$ implies $Y\not\succ X$).
\begin{thm}
A choice correspondence $\C$ is rationalizable if and only if $\succ$ is a strict partial order.
\end{thm}
\begin{proof}
Suppose that $\succ$ is a strict partial order.
Let $\hat{\succ}$ be a linear extension of $\succ$ (which can be obtained by an algorithm such as topological sorting). For each $X\in 2^I$, define
\begin{align}
f(X)=\begin{cases}
|\{Y\in\Gamma\mid X \mathbin{\hat{\succ}} Y\}| &\text{if }X\in\Gamma,\\
-1 & \text{if }X\not\in\Gamma.
\end{cases}
\end{align}
Then, it is not difficult to see that $f$ induces $\C$.

Conversely, suppose that $\succ$ is not a strict partial order.
As $\succ$ is the transitive closure of $Q$, it must fail irreflexivity or asymmetry.
In either case, we have $X\succ X$ for some $X\in\Gamma$ by transitivity.
This implies that there are sequences of subsets $X_1,X_2,\dots,X_k\in\Gamma$ and $Z_1,Z_2,\dots,Z_{k-1}\subseteq I$ such that 
\begin{itemize}
\item $k\ge 2$, 
\item $X_1=X_k$,
\item $X_i\in\C(Z_i)$ and $Z_i\supseteq X_i\cup X_{i+1}$ for all $i\in\{1,\dots,k-1\}$, and
\item $X_{i^*+1}\not\in\C(Z_{i^*})$ for some $i^*\in\{1,\dots,k-1\}$.
\end{itemize}
Suppose to the contrary that $\C$ is rationalizable by a function $u\colon 2^I\to\mathbb{R}$. Then, we have $u(X_1)\ge u(X_2)\ge\dots\ge u(X_k)=u(X_1)$ and $u(X_{i^*})>u(X_{i^*+1})$, which is a contradiction.
Thus, $\C$ is not rationalizable.
\end{proof}

\section{Relationship between PI and the conjunction of Substitutability and IRC}\label{sec:sub_PI}

In this section, we discuss the relationship between PI and the conjunction of substitutability and IRC.

\begin{lm}\label{lm:PIunion}
Suppose that a choice correspondence $\C$ can be represented as the union of choice functions $C^{(1)},\dots,C^{(k)}\colon 2^I\to 2^I$, i.e., $$\C(X)=\bigl\{C^{(1)}(X),\dots,C^{(k)}(X)\bigr\} \quad(\forall X\in 2^I).$$ 
If $C^{(1)},\dots,C^{(k)}$ are substitutable, then $\C$ is also substitutable 
Moreover, if $C^{(1)},\dots,C^{(k)}$ satisfy IRC, then $\C$ satisfies IRC.
\end{lm}
\begin{proof}
Let $X_1, X_2 \in 2^I$ with $X_1 \supseteq X_2$. 
Let $Z_1 \in \C(X_1)$, and suppose that $Z_1 = C^{(j)}(X_1)$.
If $C^{(j)}$ is substitutable, it follows that $X_2\cap Z_1=X_2 \cap C^{(j)}(Z_1) \subseteq C^{(j)}(X_2)$.
Similarly, let $Z_2\in\C(X_2)$ and suppose that $Z_2 = C^{(j)}(X_2)$.
If $C^{(j)}$ is substitutable, it follows that $X_2 \cap C^{(j)}(X_1) \subseteq C^{(j)}(X_2)=Z_2$.
Thus, any choice correspondence that can be represented as the union of substitutable choice functions satisfies  (SC$^{1}_{\text{ch}})$ and (SC$^{2}_{\text{ch}})$.

For IRC, consider $X,Y,Y'\in 2^I$ with $Y\in\C(X)$ and $Y\subseteq Y'\subseteq X$. 
Suppose that $Y=C^{(j)}(X)$. 
Then, if $C^{(j)}$ is IRC, it follows that $C^{(j)}(Y')=Y$. Thus, any choice correspondence that can be represented as the union of IRC choice functions satisfies IRC.
\end{proof}
From this lemma, any choice correspondence that can be represented as the union of PI choice functions satisfies both substitutability and IRC.

\begin{thm}\label{thm:PI-SUB}
Every PI choice correspondence satisfies substitutability and IRC.
Moreover, there is a choice correspondence $\C$ that is not PI, but satisfies substitutability and IRC.
\end{thm}
\begin{proof}
If $\C$ is a PI choice correspondence, then we have $\C(X)=\{C^w(X)\mid \text{$w$ is a UM weight}\}$ by \Cref{lm:chw}. Thus, by \Cref{lm:PIunion}, $\C$ satisfies substitutability and IRC.

Now, we consider the choice correspondence $\C_{4}$ given in 
Table \ref{tab:exChoice}. Then, we have $\C_4(X)=\{C^{(1)}(X),C^{(2)}(X)\}$ where $C^{(1)}(X)=X$ and $C^{(2)}(X)=\emptyset$ for all $X\in 2^I$.
As observed, $\C_4$ is not PI.
However, it satisfies substitutability and IRC by \Cref{lm:PIunion}.
\end{proof}

From this theorem, we can conclude that PI is a strictly stronger condition than substitutability and IRC for choice correspondences.

Finally, we examine the representation of a general upper bound in terms of choice correspondences.
A nonempty family of subsets $\cF\subseteq 2^I$ is called \emph{general upper bound} if $X\subseteq Y\in\cF$ implies $Y\in\cF$.
\begin{prop}\label{prop:general-ub}
For a general upper bound $\cF$, define a choice correspondence $\C(X)=\{Y\subseteq X\mid Y\in\cF\}$.
Then, $\C$ satisfies both substitutability and IRC.
Moreover, $\C$ can be represented as a union of PI and LAD choice functions.
\end{prop}
\begin{proof}
We have 
$$\C(X)=\{C_Y(X)\mid Y\in\cF\},$$ 
where $C_Y$ is defined by $C_Y(X)=X\cap Y$ for each $X\in 2^I$. Each $C_Y$ is a choice function that satisfies PI and LAD.
Hence, $\C$ can be represented as a union of PI and LAD choice functions.
Moreover, by \Cref{lm:PIunion}, $\C$ satisfies substitutability and IRC.
\end{proof}

\section{Constructing a Choice Oracle from a Membership Oracle}\label{sec:computation}
Let $C\colon 2^I\to 2^I$ be a PI choice function that is accessible via a membership oracle---that is, for any $X,Y\in 2^I$, we can query whether $C(X)=Y$.
In this section, we construct a choice oracle that returns $C(X)$ in polynomial time for any $X\in 2^I$, given access to the membership oracle.

If $C(X)=X$, then we are done.
Otherwise (i.e., $C(X)\ne X$), we search an element $x\in X$ that is not $C(X)$.
Once we obtain such an element $x$, it follows that $C(X)=C(X-x)$.
Thus, $C(X)$ can be computed by recursively applying the above procedure to $C(X-x)$.
Note that, by PI of $C$, if there exist $x\in X$ and $X'\subseteq X$ such that $x\in X'\setminus C(X')$, then 
$x\in X\setminus C(X)$ by $X'\setminus C(X')\subseteq X\setminus C(X)$.

We now describe a procedure to find such an element $x$.
Suppose $X=\{i_1,\dots,i_p\}$ and for each $k\in\{0,1,\dots,p\}$ define $X_k=\{i_1,\dots,i_k\}$.
Since $C(X_0)=C(\emptyset)=\emptyset=X_0$ and $C(X_p)=C(X)\ne X=X_p$, there exists an index $k^*\in\{1,\dots,p\}$ such that $C(X_{k^*-1})=X_{k^*-1}$ and $C(X_{k^*})\ne X_{k^*}$.
We can find such a $k^*$ by a linear search.
Now, if $C(X_{k^*})=X_{k^*-1}$, then $i_{k^*}$ is a desired element.
Otherwise, we have $C(X_{k^*-1})=X_{k^*-1}$, $C(X_{k^*})\ne X_{k^*-1}$, and $C(X_{k^*})\ne X_{k^*}=X_{k^*-1}+i_{k^*}$.
In this case, we must have $i_{k^*}\in C(X_{k^*})$; otherwise, by PI of $C$, we have
$$C(X_{k^*})=C(X_{k^*}\cup X_{k^*-1})=C(C(X_{k^*})\cup X_{k^*-1})=C(X_{k^*-1})=X_{k^*-1},$$
which contradicts $C(X_{k^*})\ne X_{k^*-1}$.
Define a choice function $C'\colon 2^{X_{k^*-1}}\to 2^{X_{k^*-1}}$ by $C'(X')=C(X'+i_{k^*})-i_{k^*}~(\forall X'\subseteq X_{k^*-1})$.
This choice function $C'$ satisfies PI since, for all $X',X''\subseteq I$,
\begin{align}
C'(C'(X')\cup X'')
&= C((C'(X')\cup X'')+i_{k^*})-i_{k^*}
= C(((C(X'+i_{k^*})-i_{k^*})\cup X'')+i_{k^*})-i_{k^*}\\
&= C((X'\cup X'')+i_{k^*})-i_{k^*}
= C'(X'\cup X''),
\end{align}
where the third equality uses PI of $C$.
Our task is now to find an element $x\in X'\setminus C'(X')$ for some $X'\subseteq X_{k^*-1}$ because 
$X'\setminus C'(X')=(X'+i_{k^*})\setminus C(X'+i_{k^*})$ for all $X'\subseteq X_{k^*-1}$.
We recursively apply the above procedure to $C'$ to find such an element $x$, which will output a desired element $ x$ for $X$.

The above procedure is summarized in \Cref{alg:PIchf}.

\begin{algorithm}[ht]
\SetKwFunction{Choice}{Choice}
\SetKwFunction{Discard}{Discard}
\SetKwProg{Fn}{Function}{:}{}
\tcc{Compute $C(X)$}
\Fn{$\Choice{X}$}{
    \lIf{$C(X)=X$}{\Return $X$}
    \Else{
        $x\gets \Discard(X,\emptyset)$\;
        \Return $\Choice(X-x)$\;
    }
}

\tcc{Compute $x\in (X\cup Z)\setminus C(X\cup Z)$ under the condition that $Z\subseteq C(X\cup Z)\subsetneq X\cup Z$}
\Fn{\Discard{$X,Z$}}{
    Let $X=\{i_1,\dots,i_{p-1},i_p\}$\;
    For each $k\in\{0,1,\dots,p\}$, let $X_k=\{i_1,\dots,i_k\}$\;
    Find $k^*\in\{1,\dots,p\}$ such that $C(X_{k^*-1}\cup Z)=X_{k^*-1}\cup Z$ and $C(X_{k^*}\cup Z)\ne X_{k^*}\cup Z$\;
    \lIf{$C(X_{k^*}\cup Z)=X_{k^*-1}\cup Z$}{\Return $i_{k^*}$}
    \lElse{\Return $\Discard(X_{k^*-1},\ Z+i_{k^*})$}
}
\caption{Computation of $C(X)$ for a PI choice function $C\colon 2^I\to 2^I$ and $X\in 2^I$}\label{alg:PIchf}
\end{algorithm}

\begin{thm}\label{thm:computeF}
For any PI choice function $C\colon 2^I\to 2^I$ accessible via a membership oracle and any set $X\in 2^I$, we can compute the set $C(X)$ in $O(|X|^3)$ time using \Cref{alg:PIchf}.
\end{thm}
\begin{proof}
We first show that, under the condition that $Z\subseteq C(X\cup Z)\subsetneq X\cup Z$, the procedure $\Discard(X,Z)$ outputs an element $x\in (X\cup Z)\setminus C(X\cup Z)$ in $O(|X|^2)$ time.
Since $Z\subsetneq X\cup Z$, it follows that $X\ne\emptyset$.

If $X$ is a singleton, i.e., $X=\{i_1\}$, then $\Discard(X,Z)$ returns $i_1$. 
This follows from the fact that $Z\subseteq C(Z+i_1)\subsetneq Z+i_1$ implies $C(Z+i_1)=Z$.
In this case, $i_1$ is a desired element since $i_1\in (Z+i_1)\setminus C(Z+i_1)=(X\cup Z)\setminus C(X\cup Z)$.

Otherwise, the procedure selects an element $i_{k^*}$ such that $C(X_{k^*-1}\cup Z)\setminus Z=X_{k^*-1}$ and $C(X_{k^*}\cup Z)\setminus Z\ne X_{k^*}$.
If $C(X_{k^*}\cup Z)\setminus Z=X_{k^*-1}$, then it outputs $i_{k^*}\notin X_{k^*-1}=C(X_{k^*}\cup Z)\setminus Z$.
Alternatively, if $C(X_{k^*}\cup Z)\setminus Z\ne X_{k^*-1}$, then we have $i_{k^*}\in C(X_{k^*}\cup Z)$.
Thus, we have
\begin{align}
Z+i_{k^*}\subseteq C(X_{k^*}\cup Z)
=C(X_{k^*-1}\cup (Z+i_{k^*}))\subseteq X_{k^*}\cup Z=X_{k^*-1}\cup (Z+i_{k^*}).
\end{align}
Hence, the recursive call $\Discard(X_{k^*-1},Z+i_{k^*})$ satisfies the required condition.

Since the size of the first argument decreases strictly with each call runs in $O(|X|)$ time, the total time complexity for computing $\Discard(X,Z)$ is $O(|X|^2)$.

Next, we show that $\Choice(X)$ correctly computes $C(X)$ in $O(|X|^3)$ time.
If $C(X)=X$, then it correctly outputs $X$.
Otherwise, it calls $\Discard(X,\emptyset)$ where $\emptyset\subseteq C(X)\subsetneq X$.
Thus, it obtains an element $x\in X\setminus C(X)$ in $O(|X|^2)$ time.
Since $C(X-x)=C(X)$, the function recursively calls $\Choice(X-x)$ to compute $C(X-x)$.
As the size of the argument decreases strictly with each recursive call and each call requires $O(|X|^2)$ time, the overall time complexity for computing $C(X)$ is $O(|X|^3)$.
\end{proof}

\section{Bridging}\label{sec:admissions}

In this section, we explore the relationship between our results and those presented by \citet{erdil2022corrigendum}.
They introduced the \emph{bridging} property, which applies to acceptant choice correspondences.
Recall that a choice correspondence $\C$ is called acceptant if there exists a nonnegative integer $q$ such that $|C(X)|=\min\{|X|,\,q\}$ for every $X\in 2^I$. 
We remark that acceptance is a stronger condition than LAD. 
Indeed, choice correspondences with type-specific quotas discussed in \Cref{sec:application} do not satisfy acceptance.

\begin{df}[\citet{erdil2022corrigendum}]
An acceptant choice correspondence $\C$ is said to satisfy \emph{bridging} if the following condition holds: 
Let $X,Y$ be subsets of $I$ with $Y\subseteq X$ and $|Y| \geq q$. 
Let $A \in \C(X)$ and $B \in \C(Y)$ be such that $(Y \cap A) \subseteq B$.
Then, for each $i \in A \setminus B$, there exists $j \in (B \setminus A) \cup ((X \setminus Y) \setminus A)$ such that $A -i +j \in \C(X-i)$ (see \Cref{fig:bridging}). 
\end{df}

\begin{figure}[htbp]
\centering
\begin{tikzpicture}[xscale=.9,yscale=.6]
    \draw[very thick, fill=bggreen] (-0.5,0) ellipse (3.5cm and 2cm);
    \node at (-3.5,1.6) {$X$};

    \draw[very thick,fill=white] (-.5,-.3) ellipse (2cm and 1.5cm);
    \node at (-2.5,0.5) {$Y$};
    
    \begin{scope}
        \clip (-.5,-.3) ellipse (2cm and 1.5cm); 
        \draw[very thick, draw=myblue, fill=bggreen] (0.3,-0.2) ellipse (1.5cm and 1.2cm);
    \end{scope}
    \node[myred] at (1.5,1.1) {$A$};

    \draw[very thick,rotate=10,draw=myred,fill=white] (1,-0.4) ellipse (1.5cm and 1cm);

    \draw[thick] (-.5,-.3) ellipse (2cm and 1.5cm); 
    \begin{scope}
        \clip (.3,-.2) ellipse (1.5cm and 1.2cm); 
        \draw[line width=2pt, draw=myblue, xshift=-.3mm] (-.5,-.3) ellipse (2cm and 1.5cm); 
    \end{scope}
    \node[myblue] at (-1.2,-1) {$B$};

    \node[dot,label=right:$i$] at (1.9,-.1) {};
\end{tikzpicture}
\caption{Illustration of the bridging property. The green region represent the set $(B\setminus A)\cup((X\setminus Y)\setminus A)$, which contains the element $j$ required to satisfy the condition.}
\label{fig:bridging}
\end{figure}
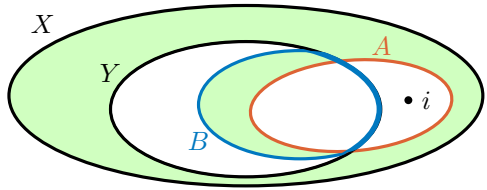

\citet{erdil2022corrigendum} demonstrated that if choice correspondences satisfy acceptance and bridging, then the necessity of Theorem~\ref{thm:main} holds.
\begin{thm}[\citet{erdil2022corrigendum}]\label{thm:bridge-efficient}
Suppose that $\C_s$ is an acceptant choice correspondence that satisfies bridging for each school $s\in S$. 
If a stable matching does not admit a PSIC, then it is constrained efficient.
\end{thm}

They also proved that if an acceptance choice correspondence fails to satisfy the bridging property, then one can construct a market in which a stable matching exists that is neither constrained efficient nor admits a PSIC.
\begin{prop}[\citet{erdil2022corrigendum}]
Suppose that $\C_{s^*}\colon 2^I\rightrightarrows 2^I$ is acceptant, but violates the bridging property.
Then, there exists a market $(I,S,(\succ_i)_{i\in I},(\C_s)_{s\in S})$ with $s^*\in S$ satisfying the following conditions:
\begin{itemize}
\item There is a stable matching that is not constrained efficient and does not admit a PSIC.
\item All schools except $s^*$ have strict responsive choice correspondences.
\end{itemize}
\end{prop}

This proposition, together with \Cref{thm:main}, implies that any PI and acceptant choice correspondence must satisfy bridging. 
We now show this directly.
\begin{prop}
Any PI and acceptant choice correspondence $\C$ satisfies bridging.
\end{prop}
\begin{proof}
Let $X,Y,A,B\in 2^I$ be $Y \subseteq X$, $|Y| \geq q$, $A \in \C(X)$, $B \in \C(Y)$, and $(Y \cap A) \subseteq B$.
Assume, without loss of generality, that there exists an element $i \in A \setminus B$; if no such $i$ exists, then the bridging condition holds trivially.

Since the acceptant property implies $|A|=|B|=q$, we have $B\setminus A\neq\emptyset$.
Let $A=\{i_1,\dots,i_p\}$, $B\setminus A=\{i_{p+1},\dots,i_q\}$, and $I\setminus (A\cup B)=\{i_q+1,\dots,i_n\}$.
Define a UM weight $w$ such that $w(i_k)=2^{n-k}$ for each $i_k\in I$.
By the construction of $w$, we have $A\subseteq C^{w}(X)$. 
Since acceptance ensures $|C^w(X)|=q$, we have $A= C^{w}(X)$.
Similarly, by the construction of $w$, we have $B\subseteq C^{w}(Y)$.
Since acceptance implies $|C^w(Y)|=q$, we have $B= C^{w}(Y)$.

Since $Y\subseteq X-i$, we have $|X-i|\ge |Y|\ge q$. 
Thus, there exists $j\in X\setminus A$ such that $A -i +j = C^w(X-i)$. 
We prove that $j \in (B \setminus A) \cup ((X \setminus Y) \setminus A)$ by contradiction. 
Suppose, toward a contradiction, that $j\in Y\setminus B$.
By PI of $C^{w}$ and $B=C^w(Y)$, we have 
\begin{align}
j\in Y\setminus B=Y\setminus C^{w}(Y)\subseteq (X-i)\setminus C^{w}(X-i). 
\end{align}
This contradicts $A -i +j = C^w(X-i)$.
Therefore, we have $j \in (B \setminus A) \cup ((X \setminus Y) \setminus A)$, which completes the proof.
\end{proof}

From this proposition, we conclude that the combination of the PI and acceptant conditions is weaker than the bridging property. 
Additionally, we will show that an acceptant choice correspondence satisfying bridging need not satisfy PI. 
Moreover, even for a market in which every school employs an acceptant choice correspondence that satisfies bridging, a constrained efficient stable matching admitting a PSIC may still exist. 
We illustrate these facts in the following subsections.

\subsection{Admissions by a Committee}

Suppose that a school has $q\in\mathbb{Z}_+$ seats to fill.
Let $H$ be a set of referees, each of whom $h\in H$ has a strict order $\succ_h$ over the set of students $I$.
A function $\pi\colon\{1,\dots,q\}\to H$ induces a choice function $C^{\pi}$ as follows:
\begin{align}
C^{\pi}(X)=\{i_1,\dots,i_{\min\{q,|X|\}}\} \quad(\forall X\in 2^I),
\end{align}
where $\{i_1\}=\argmax_{\succ_{\pi(1)}}X$, and $i_\ell=\argmax_{\succ_{\pi(\ell)}}X\setminus\{i_1,\dots,i_{\ell-1}\}$ for $\ell=2,\dots,\min\{q,|X|\}$.
Let $\Pi$ be the set of all functions $\pi$. 
A choice correspondence induced by admissions by a committee is defined as 
\begin{align}
\C^H(X)=\{C^{\pi}(X) \mid \pi\in\Pi\}\quad(\forall X\in 2^I).
\end{align}

Clearly, $\mathcal{C}^H$ is acceptant.
\citet{erdil2019efficiency} showed that $\mathcal{C}^H$ satisfies substitutability. 
Moreover, they demonstrated that it also satisfies bridging.

Furthermore, \citet{erdil2019efficiency} and \citet{erdil2022corrigendum} proved that no PSIC is a necessary and sufficient condition to be constrained efficient if every school has a choice correspondence induced by admissions by a committee.
\begin{prop}[\citet{erdil2019efficiency,erdil2022corrigendum}]
Suppose every school $s\in S$ has a choice correspondence induced by admissions by a committee. 
A stable matching is constrained efficient if and only if it does not admit a PSIC.
\end{prop}

However, a choice correspondence induced by admissions by a committee may not satisfy PI.
\begin{ex}\label{ex:admission}
Let $I=\{i_1,i_2,i_3,i_4\}$ and $q=2$.
Suppose that the set of referees is $H=\{h_1,h_2,h_3\}$, where each referee has a strict preference order given by
\begin{align}
{\succ_{h_1}}=(i_1\ i_2\ i_3\ i_4),\quad
{\succ_{h_2}}=(i_1\ i_3\ i_2\ i_4),\quad
{\succ_{h_3}}=(i_2\ i_4\ i_1\ i_3).
\end{align}
Then, $\C^H(\{i_1,i_2,i_3,i_4\})=\big\{\{i_1,i_2\},\{i_1,i_3\},\{i_2,i_4\}\big\}$ is not a g-matroid.
Moreover, $\C^H$ is not rationalizable because
$\C^H(\{i_1,i_2,i_3,i_4\})=\big\{\{i_1,i_2\},\{i_1,i_3\},\{i_2,i_4\}\big\}$ implies that $\{i_2,i_3\}$ is strictly worse than $\{i_2,i_4\}$, 
while $\C^H(\{i_2,i_3,i_4\})=\big\{\{i_2,i_3\},\{i_2,i_4\}\big\}$ implies that $\{i_2,i_3\}$ and $\{i_2,i_4\}$ are equally valuable.
\end{ex}

\subsection{Restricted Admissions by a committee}

In the admissions by a committee framework described in the previous subsection, all possible combinations of referee assignments are considered.
However, this generality may not align with certain practical scenarios where only specific subsets of referee assignments are relevant or permissible. To address this, we consider a restricted setting in which the set of possible functions $\pi\colon\{1,\dots,q\}\to H$ is limited to a subset of interest.

Formally, let $\Pi$ be a subset of $\Pi$, where $\Pi$ is the set of all functions from $\{1,\dots,q\}$ to $H$.
A choice correspondence induced by admissions by a committee in this restricted setting is defined as 
\begin{align}
\C^{H,\Pi'}(X)=\{C^{\pi}(X) \mid \pi\in\Pi'\}\quad(\forall X\in 2^I).
\end{align}
where $\Pi'$ is the set of all injective functions from $\{1,\dots,q\}$ to $H$.
This restriction allows us to model practical scenarios better while maintaining flexibility.

It is not difficult to see that a choice correspondence induced by restricted admissions by a committee is substitutable and acceptant.
Whether the bridging property is satisfied depends on the particular instance.
However, the following example illustrates that even when the bridging property holds, there exists a market in which a constrained efficient stable matching admits a PSIC.
\begin{ex}\label{ex:bridging-nosufficient}
Let $I=\{i_1,i_2,i_3,i_4,i_5\}$ and $S=\{s_1,s_2,s_3,s_4\}$.
Assume that the preferences of students are given by
\begin{gather}
{\succ_{i_1}}=(s_2\ s_1\ \vn\ s_3\ s_4),\quad
{\succ_{i_2}}=(s_3\ s_1\ \vn\ s_2\ s_4),\quad
{\succ_{i_3}}=(s_1\ s_2\ \vn\ s_3\ s_4),\\[1ex]
{\succ_{i_4}}=(s_1\ s_3\ \vn\ s_2\ s_4),\quad
{\succ_{i_5}}=(s_1\ s_4\ \vn\ s_2\ s_3).
\end{gather}
The choice correspondences for $s_2$, $s_3$, and $s_4$ are given as 
\begin{align}
\C_{s_2}(X)=\C_{s_3}(X)=\C_{s_4}(X)
=\argmax\{|Y|\mid Y\subseteq X,\ |Y|\le 1\}\quad(\forall X\in 2^I).
\end{align}
Note that these can be represented as admissions by a committee.
Suppose that the set of referees for $s_1$ is $H=\{h_1,h_2,h_3\}$, where each referee has a strict preference order given by
\begin{align}
{\succ_{h_1}}=(i_1\ i_2\ i_3\ i_5\ i_4),\quad
{\succ_{h_2}}=(i_1\ i_3\ i_2\ i_5\ i_4),\quad
{\succ_{h_3}}=(i_2\ i_4\ i_1\ i_5\ i_3).
\end{align}
Let $q=2$ and consider the restricted setting where the same referee is selected, i.e., $\Pi'=\{(h_1,h_1),(h_2,h_2),(h_3,h_3)\}$.
Then, the choice correspondence for $s_1$ is defined by $\C_{s_1}\equiv\C^{H,\Pi'}$. 
We can verify that $\C_{s_1}$ satisfies bridging by enumerating all the possible combinations of $X,Y,A,B\in 2^I$ and $i\in A\setminus B$ such that $Y \subseteq X$, $|Y| \geq q$, $A \in \C_{s_1}(X)$, $B \in \C_{s_1}(Y)$, and $(Y \cap A) \subseteq B$ (see \Cref{tab:exist-j}).

In this market, the matching $\mu=\{(i_1,s_1),(i_2,s_1),(i_3,s_2),(i_4,s_3),(i_5,s_4)\}$ is a constrained efficient stable matching.
For this matching $\mu$, there is a unique PSIC $(i_1,i_3,i_2,i_4)$ (see \Cref{fig:efficientPSIC}).

\begin{figure}[thbp]
\centering
\begin{tikzpicture}[xscale=2,yscale=1,font=\scriptsize,line width=3pt]
\draw[opacity=0] (-.5,0.5) grid[step=0.5] (2,5.5);
\node[dot=2mm,label=left:$(\textcolor{myblue}{s_2}\ \textcolor{myred}{s_1}\ \vn\ s_3\ s_4)\quad i_1$] (i1) at (0,5) {};
\node[dot=2mm,label=left:$(\textcolor{myblue}{s_3}\ \textcolor{myred}{s_1}\ \vn\ s_2\ s_4)\quad i_2$] (i2) at (0,4) {};
\node[dot=2mm,label=left:$(\textcolor{myblue}{s_1}\ \textcolor{myred}{s_2}\ \vn\ s_3\ s_4)\quad i_3$] (i3) at (0,3) {};
\node[dot=2mm,label=left:$(\textcolor{myblue}{s_1}\ \textcolor{myred}{s_3}\ \vn\ s_2\ s_4)\quad i_4$] (i4) at (0,2) {};
\node[dot=2mm,label=left:$(\textcolor{myblue}{s_1}\ \textcolor{myred}{s_4}\ \vn\ s_3\ s_4)\quad i_5$] (i5) at (0,1) {};
\node[dot=2mm,label=right:{$s_1$}] (s1) at (1,4.5) {};
\node[dot=2mm,label=right:{$s_2$}] (s2) at (1,3) {};
\node[dot=2mm,label=right:{$s_3$}] (s3) at (1,2) {};
\node[dot=2mm,label=right:{$s_4$}] (s4) at (1,1) {};
\node[text width=2cm] at (1.8,4.5) {\textcolor{myred}{$\{i_1,i_2\}$}\\[1mm]$\{i_1,i_3\}$\\[1mm]$\{i_2,i_4\}$};

\foreach \x/\y in {i1/s2,i2/s3,i3/s1,i4/s1,i5/s1} {
    \draw[bgblue,dotted] (\x) -- (\y);
}
\foreach \x/\y in {i1/s1,i2/s1,i3/s2,i4/s3,i5/s4} {
    \draw[myred] (\x) -- (\y);
}

\draw[green,line width=1.2pt,rounded corners=5pt] (i3.center)--($(s1)$)--(i2.center)--(s3.center)--(i4.center)--($(s1)+(2pt,1pt)$)--(i1.center)--(s2.center)--cycle;
\draw[mygreen,opacity=.5,line width=.8pt,rounded corners=5pt] (i3.center)--($(s1)$)--(i2.center)--(s3.center)--(i4.center)--($(s1)+(2pt,1pt)$)--(i1.center)--(s2.center)--cycle;
\draw[green!20!black,line width=0.2pt,rounded corners=5pt] (i3.center)--($(s1)$)--(i2.center)--(s3.center)--(i4.center)--($(s1)+(2pt,1pt)$)--(i1.center)--(s2.center)--cycle;

\end{tikzpicture}
\caption{Stable matching $\mu$ (red) and the unique PSIC (green)}\label{fig:efficientPSIC}

\end{figure}
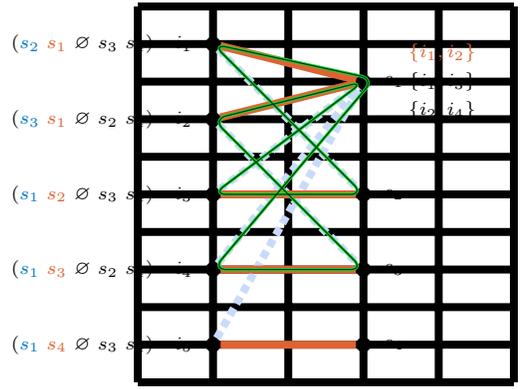
\end{ex}

\begin{table}[p]
\centering
\scalebox{.55}{
\begin{tabular}{cccccc}
\toprule
$X$ & $Y$ & $A$ & $B$ & $i$ & $j$\\\midrule
$\{i_1,i_2,i_3,i_4\}$ & $\{i_1,i_2,i_3\}$ & $\{i_2,i_4\}$ & $\{i_1,i_2\}$ & $i_4$ & $i_1$\\
$\{i_1,i_2,i_3,i_4\}$ & $\{i_1,i_2,i_4\}$ & $\{i_1,i_3\}$ & $\{i_1,i_2\}$ & $i_3$ & $i_2$\\
$\{i_1,i_2,i_3,i_4\}$ & $\{i_1,i_3,i_4\}$ & $\{i_2,i_4\}$ & $\{i_1,i_4\}$ & $i_2$ & $i_1$\\
$\{i_1,i_2,i_3,i_4\}$ & $\{i_1,i_3,i_4\}$ & $\{i_1,i_2\}$ & $\{i_1,i_4\}$ & $i_2$ & $i_4$\\
$\{i_1,i_2,i_3,i_4\}$ & $\{i_1,i_3,i_4\}$ & $\{i_1,i_2\}$ & $\{i_1,i_3\}$ & $i_2$ & $i_3$\\
$\{i_1,i_2,i_3,i_4\}$ & $\{i_2,i_3,i_4\}$ & $\{i_1,i_3\}$ & $\{i_2,i_3\}$ & $i_1$ & $i_2$\\
$\{i_1,i_2,i_3,i_4\}$ & $\{i_2,i_3,i_4\}$ & $\{i_1,i_2\}$ & $\{i_2,i_4\}$ & $i_1$ & $i_4$\\
$\{i_1,i_2,i_3,i_4\}$ & $\{i_2,i_3,i_4\}$ & $\{i_1,i_2\}$ & $\{i_2,i_3\}$ & $i_1$ & $i_3$\\
$\{i_1,i_2,i_3,i_5\}$ & $\{i_1,i_2,i_5\}$ & $\{i_1,i_3\}$ & $\{i_1,i_2\}$ & $i_3$ & $i_2$\\
$\{i_1,i_2,i_3,i_5\}$ & $\{i_1,i_3,i_5\}$ & $\{i_1,i_2\}$ & $\{i_1,i_5\}$ & $i_2$ & $i_5$\\
$\{i_1,i_2,i_3,i_5\}$ & $\{i_1,i_3,i_5\}$ & $\{i_1,i_2\}$ & $\{i_1,i_3\}$ & $i_2$ & $i_3$\\
$\{i_1,i_2,i_3,i_5\}$ & $\{i_2,i_3,i_5\}$ & $\{i_1,i_3\}$ & $\{i_2,i_3\}$ & $i_1$ & $i_2$\\
$\{i_1,i_2,i_3,i_5\}$ & $\{i_2,i_3,i_5\}$ & $\{i_1,i_2\}$ & $\{i_2,i_5\}$ & $i_1$ & $i_5$\\
$\{i_1,i_2,i_3,i_5\}$ & $\{i_2,i_3,i_5\}$ & $\{i_1,i_2\}$ & $\{i_2,i_3\}$ & $i_1$ & $i_3$\\
$\{i_1,i_2,i_4,i_5\}$ & $\{i_1,i_2,i_5\}$ & $\{i_2,i_4\}$ & $\{i_1,i_2\}$ & $i_4$ & $i_1$\\
$\{i_1,i_2,i_4,i_5\}$ & $\{i_1,i_4,i_5\}$ & $\{i_2,i_4\}$ & $\{i_1,i_4\}$ & $i_2$ & $i_1$\\
$\{i_1,i_2,i_4,i_5\}$ & $\{i_1,i_4,i_5\}$ & $\{i_1,i_2\}$ & $\{i_1,i_4\}$ & $i_2$ & $i_4$\\
$\{i_1,i_2,i_4,i_5\}$ & $\{i_1,i_4,i_5\}$ & $\{i_1,i_2\}$ & $\{i_1,i_5\}$ & $i_2$ & $i_5$\\
$\{i_1,i_2,i_4,i_5\}$ & $\{i_2,i_4,i_5\}$ & $\{i_1,i_2\}$ & $\{i_2,i_4\}$ & $i_1$ & $i_4$\\
$\{i_1,i_2,i_4,i_5\}$ & $\{i_2,i_4,i_5\}$ & $\{i_1,i_2\}$ & $\{i_2,i_5\}$ & $i_1$ & $i_5$\\
$\{i_1,i_3,i_4,i_5\}$ & $\{i_1,i_3,i_5\}$ & $\{i_1,i_4\}$ & $\{i_1,i_5\}$ & $i_4$ & $i_5$\\
$\{i_1,i_3,i_4,i_5\}$ & $\{i_1,i_3,i_5\}$ & $\{i_1,i_4\}$ & $\{i_1,i_3\}$ & $i_4$ & $i_3$\\
$\{i_1,i_3,i_4,i_5\}$ & $\{i_1,i_4,i_5\}$ & $\{i_1,i_3\}$ & $\{i_1,i_4\}$ & $i_3$ & $i_4$\\
$\{i_1,i_3,i_4,i_5\}$ & $\{i_1,i_4,i_5\}$ & $\{i_1,i_3\}$ & $\{i_1,i_5\}$ & $i_3$ & $i_5$\\
$\{i_1,i_3,i_4,i_5\}$ & $\{i_3,i_4,i_5\}$ & $\{i_1,i_4\}$ & $\{i_4,i_5\}$ & $i_1$ & $i_5$\\
$\{i_1,i_3,i_4,i_5\}$ & $\{i_3,i_4,i_5\}$ & $\{i_1,i_3\}$ & $\{i_3,i_5\}$ & $i_1$ & $i_5$\\
$\{i_2,i_3,i_4,i_5\}$ & $\{i_2,i_3,i_5\}$ & $\{i_2,i_4\}$ & $\{i_2,i_5\}$ & $i_4$ & $i_5$\\
$\{i_2,i_3,i_4,i_5\}$ & $\{i_2,i_3,i_5\}$ & $\{i_2,i_4\}$ & $\{i_2,i_3\}$ & $i_4$ & $i_3$\\
$\{i_2,i_3,i_4,i_5\}$ & $\{i_2,i_4,i_5\}$ & $\{i_2,i_3\}$ & $\{i_2,i_4\}$ & $i_3$ & $i_4$\\
$\{i_2,i_3,i_4,i_5\}$ & $\{i_2,i_4,i_5\}$ & $\{i_2,i_3\}$ & $\{i_2,i_5\}$ & $i_3$ & $i_5$\\
$\{i_2,i_3,i_4,i_5\}$ & $\{i_3,i_4,i_5\}$ & $\{i_2,i_4\}$ & $\{i_4,i_5\}$ & $i_2$ & $i_5$\\
$\{i_2,i_3,i_4,i_5\}$ & $\{i_3,i_4,i_5\}$ & $\{i_2,i_3\}$ & $\{i_3,i_5\}$ & $i_2$ & $i_5$\\
$\{i_1,i_2,i_3,i_4,i_5\}$ & $\{i_1,i_2,i_3\}$ & $\{i_2,i_4\}$ & $\{i_1,i_2\}$ & $i_4$ & $i_1$\\
$\{i_1,i_2,i_3,i_4,i_5\}$ & $\{i_1,i_2,i_4\}$ & $\{i_1,i_3\}$ & $\{i_1,i_2\}$ & $i_3$ & $i_2$\\
$\{i_1,i_2,i_3,i_4,i_5\}$ & $\{i_1,i_2,i_5\}$ & $\{i_2,i_4\}$ & $\{i_1,i_2\}$ & $i_4$ & $i_1$\\
$\{i_1,i_2,i_3,i_4,i_5\}$ & $\{i_1,i_2,i_5\}$ & $\{i_1,i_3\}$ & $\{i_1,i_2\}$ & $i_3$ & $i_2$\\
$\{i_1,i_2,i_3,i_4,i_5\}$ & $\{i_1,i_3,i_4\}$ & $\{i_2,i_4\}$ & $\{i_1,i_4\}$ & $i_2$ & $i_1$\\
$\{i_1,i_2,i_3,i_4,i_5\}$ & $\{i_1,i_3,i_4\}$ & $\{i_1,i_2\}$ & $\{i_1,i_4\}$ & $i_2$ & $i_4$\\
$\{i_1,i_2,i_3,i_4,i_5\}$ & $\{i_1,i_3,i_4\}$ & $\{i_1,i_2\}$ & $\{i_1,i_3\}$ & $i_2$ & $i_3$\\
$\{i_1,i_2,i_3,i_4,i_5\}$ & $\{i_1,i_3,i_5\}$ & $\{i_2,i_4\}$ & $\{i_1,i_5\}$ & $i_2$ & $i_1$\\
$\{i_1,i_2,i_3,i_4,i_5\}$ & $\{i_1,i_3,i_5\}$ & $\{i_2,i_4\}$ & $\{i_1,i_5\}$ & $i_4$ & $i_1$\\
$\{i_1,i_2,i_3,i_4,i_5\}$ & $\{i_1,i_3,i_5\}$ & $\{i_2,i_4\}$ & $\{i_1,i_3\}$ & $i_2$ & $i_1$\\
$\{i_1,i_2,i_3,i_4,i_5\}$ & $\{i_1,i_3,i_5\}$ & $\{i_2,i_4\}$ & $\{i_1,i_3\}$ & $i_4$ & $i_1$\\
$\{i_1,i_2,i_3,i_4,i_5\}$ & $\{i_1,i_3,i_5\}$ & $\{i_1,i_2\}$ & $\{i_1,i_5\}$ & $i_2$ & $i_4$\\
$\{i_1,i_2,i_3,i_4,i_5\}$ & $\{i_1,i_3,i_5\}$ & $\{i_1,i_2\}$ & $\{i_1,i_3\}$ & $i_2$ & $i_3$\\
$\{i_1,i_2,i_3,i_4,i_5\}$ & $\{i_1,i_4,i_5\}$ & $\{i_2,i_4\}$ & $\{i_1,i_4\}$ & $i_2$ & $i_1$\\
$\{i_1,i_2,i_3,i_4,i_5\}$ & $\{i_1,i_4,i_5\}$ & $\{i_1,i_3\}$ & $\{i_1,i_4\}$ & $i_3$ & $i_2$\\
$\{i_1,i_2,i_3,i_4,i_5\}$ & $\{i_1,i_4,i_5\}$ & $\{i_1,i_3\}$ & $\{i_1,i_5\}$ & $i_3$ & $i_2$\\
$\{i_1,i_2,i_3,i_4,i_5\}$ & $\{i_1,i_4,i_5\}$ & $\{i_1,i_2\}$ & $\{i_1,i_4\}$ & $i_2$ & $i_3$\\
$\{i_1,i_2,i_3,i_4,i_5\}$ & $\{i_1,i_4,i_5\}$ & $\{i_1,i_2\}$ & $\{i_1,i_5\}$ & $i_2$ & $i_3$\\
$\{i_1,i_2,i_3,i_4,i_5\}$ & $\{i_2,i_3,i_4\}$ & $\{i_1,i_3\}$ & $\{i_2,i_3\}$ & $i_1$ & $i_2$\\
$\{i_1,i_2,i_3,i_4,i_5\}$ & $\{i_2,i_3,i_4\}$ & $\{i_1,i_2\}$ & $\{i_2,i_4\}$ & $i_1$ & $i_4$\\
$\{i_1,i_2,i_3,i_4,i_5\}$ & $\{i_2,i_3,i_4\}$ & $\{i_1,i_2\}$ & $\{i_2,i_3\}$ & $i_1$ & $i_3$\\
$\{i_1,i_2,i_3,i_4,i_5\}$ & $\{i_2,i_3,i_5\}$ & $\{i_2,i_4\}$ & $\{i_2,i_5\}$ & $i_4$ & $i_1$\\
$\{i_1,i_2,i_3,i_4,i_5\}$ & $\{i_2,i_3,i_5\}$ & $\{i_2,i_4\}$ & $\{i_2,i_3\}$ & $i_4$ & $i_1$\\
$\{i_1,i_2,i_3,i_4,i_5\}$ & $\{i_2,i_3,i_5\}$ & $\{i_1,i_3\}$ & $\{i_2,i_3\}$ & $i_1$ & $i_2$\\
$\{i_1,i_2,i_3,i_4,i_5\}$ & $\{i_2,i_3,i_5\}$ & $\{i_1,i_2\}$ & $\{i_2,i_5\}$ & $i_1$ & $i_4$\\
$\{i_1,i_2,i_3,i_4,i_5\}$ & $\{i_2,i_3,i_5\}$ & $\{i_1,i_2\}$ & $\{i_2,i_3\}$ & $i_1$ & $i_3$\\
$\{i_1,i_2,i_3,i_4,i_5\}$ & $\{i_2,i_4,i_5\}$ & $\{i_1,i_3\}$ & $\{i_2,i_4\}$ & $i_1$ & $i_2$\\
$\{i_1,i_2,i_3,i_4,i_5\}$ & $\{i_2,i_4,i_5\}$ & $\{i_1,i_3\}$ & $\{i_2,i_4\}$ & $i_3$ & $i_2$\\
$\{i_1,i_2,i_3,i_4,i_5\}$ & $\{i_2,i_4,i_5\}$ & $\{i_1,i_3\}$ & $\{i_2,i_5\}$ & $i_1$ & $i_2$\\
$\{i_1,i_2,i_3,i_4,i_5\}$ & $\{i_2,i_4,i_5\}$ & $\{i_1,i_3\}$ & $\{i_2,i_5\}$ & $i_3$ & $i_2$\\
$\{i_1,i_2,i_3,i_4,i_5\}$ & $\{i_2,i_4,i_5\}$ & $\{i_1,i_2\}$ & $\{i_2,i_4\}$ & $i_1$ & $i_3$\\
$\{i_1,i_2,i_3,i_4,i_5\}$ & $\{i_2,i_4,i_5\}$ & $\{i_1,i_2\}$ & $\{i_2,i_5\}$ & $i_1$ & $i_3$\\
$\{i_1,i_2,i_3,i_4,i_5\}$ & $\{i_3,i_4,i_5\}$ & $\{i_2,i_4\}$ & $\{i_4,i_5\}$ & $i_2$ & $i_1$\\
$\{i_1,i_2,i_3,i_4,i_5\}$ & $\{i_3,i_4,i_5\}$ & $\{i_1,i_3\}$ & $\{i_3,i_5\}$ & $i_1$ & $i_2$\\
$\{i_1,i_2,i_3,i_4,i_5\}$ & $\{i_3,i_4,i_5\}$ & $\{i_1,i_2\}$ & $\{i_4,i_5\}$ & $i_1$ & $i_4$\\
$\{i_1,i_2,i_3,i_4,i_5\}$ & $\{i_3,i_4,i_5\}$ & $\{i_1,i_2\}$ & $\{i_4,i_5\}$ & $i_2$ & $i_4$\\
$\{i_1,i_2,i_3,i_4,i_5\}$ & $\{i_3,i_4,i_5\}$ & $\{i_1,i_2\}$ & $\{i_3,i_5\}$ & $i_1$ & $i_3$\\
$\{i_1,i_2,i_3,i_4,i_5\}$ & $\{i_3,i_4,i_5\}$ & $\{i_1,i_2\}$ & $\{i_3,i_5\}$ & $i_2$ & $i_3$\\
$\{i_1,i_2,i_3,i_4,i_5\}$ & $\{i_1,i_2,i_3,i_5\}$ & $\{i_2,i_4\}$ & $\{i_1,i_2\}$ & $i_4$ & $i_1$\\
$\{i_1,i_2,i_3,i_4,i_5\}$ & $\{i_1,i_2,i_4,i_5\}$ & $\{i_1,i_3\}$ & $\{i_1,i_2\}$ & $i_3$ & $i_2$\\
$\{i_1,i_2,i_3,i_4,i_5\}$ & $\{i_1,i_3,i_4,i_5\}$ & $\{i_2,i_4\}$ & $\{i_1,i_4\}$ & $i_2$ & $i_1$\\
$\{i_1,i_2,i_3,i_4,i_5\}$ & $\{i_1,i_3,i_4,i_5\}$ & $\{i_1,i_2\}$ & $\{i_1,i_4\}$ & $i_2$ & $i_4$\\
$\{i_1,i_2,i_3,i_4,i_5\}$ & $\{i_1,i_3,i_4,i_5\}$ & $\{i_1,i_2\}$ & $\{i_1,i_3\}$ & $i_2$ & $i_3$\\
$\{i_1,i_2,i_3,i_4,i_5\}$ & $\{i_2,i_3,i_4,i_5\}$ & $\{i_1,i_3\}$ & $\{i_2,i_3\}$ & $i_1$ & $i_2$\\
$\{i_1,i_2,i_3,i_4,i_5\}$ & $\{i_2,i_3,i_4,i_5\}$ & $\{i_1,i_2\}$ & $\{i_2,i_4\}$ & $i_1$ & $i_4$\\
$\{i_1,i_2,i_3,i_4,i_5\}$ & $\{i_2,i_3,i_4,i_5\}$ & $\{i_1,i_2\}$ & $\{i_2,i_3\}$ & $i_1$ & $i_3$\\
\bottomrule
\end{tabular}}
\caption{Existence of $j\in (B\setminus A)\cup((X\setminus Y)\setminus A)$ for every possible combination of $(X,Y,A,B,i)$ with $|Y|>2$ and $X\supsetneq Y$}\label{tab:exist-j}
\end{table}

\end{document}